\documentclass[12pt,english]{article}
\usepackage[english]{babel}
\usepackage[dvips]{graphicx}
\makeatletter
\usepackage{cite}
\usepackage{paralist}
\usepackage[cmex10]{amsmath}
\usepackage{amssymb}
\usepackage{amsthm}
\usepackage{psfrag}
\usepackage{dsfont}
\usepackage{multirow}
\usepackage[dvips]{graphicx}
\usepackage{subfig}

\newtheorem{theorem}{Theorem}[section]

\def\E{\mathbb{E}}

\begin{document}

\title{Optimizing the Medium Access Control in Multi-hop Wireless Networks}

\author{
Salman Malik\footnote{INRIA, Paris-Rocquencourt, France. Email: \texttt{salman.malik@inria.fr}}, 
Philippe Jacquet\footnote{Alcatel Lucent Bell Labs, Villarceaux, France. Email: \texttt{philippe.jacquet@alcatel-lucent.fr}}, 
Cedric Adjih\footnote{INRIA, Paris-Rocquencourt, France. Email: \texttt{cedric.adjih@inria.fr}}
}

\date{}

\maketitle

\begin{abstract}

We study the problem of geometric optimization of medium access control in multi-hop wireless network. We discuss the optimal placements of simultaneous transmitters in the network and our general framework allows us to evaluate the performance gains of highly managed medium access control schemes that would be required to implement these placements. In a wireless network consisting of randomly distributed nodes, our performance metrics are the optimum transmission range that achieves the most optimal tradeoff between the progress of packets in desired directions towards their respective destinations and the total number of transmissions required to transport packets to their destinations. We evaluate ALOHA based scheme where simultaneous transmitters are dispatched according to a uniform Poisson distribution and compare it with various grid pattern based schemes where simultaneous transmitters are positioned in specific regular patterns. Our results show that optimizing the medium access control in multi-hop network should take into account the parameters like signal-to-interference ratio threshold and attenuation coefficient. For instance, at {\em typical} values of signal-to-interference ratio threshold and attenuation coefficient, the most optimal scheme is based on triangular grid pattern and, under {\em no fading} channel model, the most optimal transmission range and network capacity are higher than the optimum transmission range and capacity achievable with ALOHA based scheme by factors of {\em two} and {\em three} respectively. Later on, we also identify the optimal medium access control schemes when signal-to-interference ratio threshold and attenuation coefficient approach the extreme values and discuss how fading impacts the performance of all schemes we evaluate in this article. 

\end{abstract}

\section{Introduction}

The goal of designing medium access control (MAC) schemes or protocols for wireless networks is to allow nodes to efficiently share the medium. Simultaneous transmitters should be selected in such a way that improves the {\em spatial reuse} in the network, {\it i.e.}, the shared medium should be used simultaneously by transmitters which are spatially separated in such a way that they do not interfere with each other's transmissions. In other words, permanent spatial clustering of simultaneous transmitters should never occur in a wireless network. In order to achieve this goal, MAC schemes generally use an exclusion rule to prevent nodes which are located near to each other from transmitting simultaneously.

The task of designing MAC schemes can be a challenging problem because of the difficulties posed by various types of wireless networks. To name a few: dynamic network topology, lack of any centralized control, energy constraints, channel characteristics, mobility of nodes, radio hardware limits, {\it etc.} are some of the issues faced by network and protocol designers. These challenges influence the design and protocol overhead of MAC schemes in various ways, {\it e.g.}:
\begin{itemize}
\item[--] protocol overhead to implement MAC schemes in various types of wireless networks, {\it e.g.}, intercommunication required between nodes for the distributed implementation of the protocol,
\item[--] scaling properties of the MAC scheme when the number of nodes in the network increases,
\item[--] energy efficiency, {\it e.g.}, minimizing the number of transmissions required for the functioning of the protocol, {\it etc.}
\end{itemize}

From the performance point of view, the question of what is the most optimal MAC scheme in multi-hop wireless network is extremely challenging to answer. Moreover, the implementation of this MAC scheme in a realistic wireless network may also require complicated protocol overhead. Therefore, an important issue in the design and implementation of MAC schemes is to understand the tradeoff between performance and protocol overhead costs. 

In this work, we will try to identify the most optimal placements of simultaneous transmitters in a wireless network and we will evaluate various MAC schemes that will implement these placements. These schemes represent particular points on the performance versus protocol overhead tradeoff curve and our analysis will provide useful insights for network and protocol designers. This article is organized in the following sections. In \S \ref{sec:related}, we will discuss some of the important related works. The network and transmission models are presented in \S \ref{range_sec:model} whereas the models of all MAC schemes are discussed in \S \ref{range_sec:mac_model}. We will discuss our performance metrics and also set a general framework to evaluate the performance of these MAC schemes in \S \ref{sec_performance_analysis}. The results of our evaluation are presented in \S \ref{range_sec:results}. We will also discuss the impact of the extreme values of signal-to-interference ratio (SIR) threshold and attenuation coefficient on the optimum transmission range of grid pattern based schemes in \S \ref{sec:asymp}. The impact of fading on the performance of MAC schemes will be discussed in \S \ref{sec:impact_fading} and the article concludes with insights and perspectives in \S \ref{sec:conclude}.

\section{Related Works}
\label{sec:related}

In one of the first analyses, the article \cite{Nelson:Kleinrock} studied the capacity of wireless network with slotted ALOHA and despite using a very simple geometric model for the reception of signals, the result is very similar to what can be obtained under a realistic signal-to-interference plus noise ratio (SINR) based interference model where a signal is received only if its SINR is above a desired threshold. Under a similar geometric model and assuming that all nodes are within range of each other, the article \cite{CSMA} evaluated carrier sense multiple access (CSMA) scheme and compared it with slotted ALOHA in terms of throughput (bit-rate). The article \cite{Bartek} used simulations to analyze CSMA and compared it with slotted and un-slotted ALOHA. For simulations, the authors of this article assumed Poisson distributed transmitters with density $\lambda$. Each transmitter sends packets to its assigned receiver located at a fixed distance of $a\sqrt{\lambda}$, for $a>0$. 

The authors of the article \cite{malik11} introduced the concept of {\em local capacity} for single-hop networks and defined it as the average information rate received by a node randomly located in the network. They showed that the local capacity in wireless network is optimized by the scheme based on triangular grid pattern and slotted ALOHA can achieve at least half of this optimal local capacity. They also showed that node coloring and CSMA can achieve almost the same local capacity as the most optimal scheme, by negligible difference. 

The articles \cite{Weber,Weber2} introduced the concept of {\em transmission capacity} which is defined as the maximum number of successful transmissions per unit area at a specified {\em outage probability}, {\it i.e.}, the probability that the SINR at the receiver is below a certain threshold. The authors of these articles studied the transmission capacity of ALOHA and code division multiple access (CDMA) schemes under the assumption that simultaneous transmitters form a homogeneous Poisson point process and, as in the article \cite{Bartek}, they also assumed that the receivers are located at fixed distance from their transmitters. The fact that the receivers are not a part of the network model and are located at a fixed distance from their transmitters is a simplification. An accurate model of wireless network should consider that the transmitters, transmit to receivers which are randomly located within their neighborhood. The article \cite{Weber} also analyzed the case where receivers are located at a random distance but it was assumed that the transmitters employ transmit power control such that the signal power at their intended receivers is some fixed value. 

Poisson point process can only accurately model an ALOHA based scheme where transmitters are independently and uniformly distributed in the network area. However, with MAC schemes based on exclusion rules, like node coloring or CSMA, modeling simultaneous transmitters as Poisson point process leads to an inaccurate representation of the distribution of signal level of interference. On the other hand, correlation between the location of simultaneous transmitters makes it extremely difficult to develop a tractable analytical model and derive the closed-form expression for the distribution of interference and probability of successful reception. Some of the proposed approaches are as follows. In Chapter $18$ of the book \cite{CSMA-Model}, the authors used a Mat\'ern point process for CSMA scheme whereas the article \cite{Busson} proposed to use Simple Sequential Inhibition ({\em SSI}) or an extension of {\em SSI} called {\em SSI$_k$} point process. In articles \cite{Guard,Guard2}, simultaneous interferers are modeled as Poisson point process and outage probability is obtained by excluding or suppressing some of the interferers in the guard zone around a receiver. The article \cite{Weber3} analyzed transmission capacity in networks with {\em general} node distributions under a restrictive hypothesis that density of interferers is very low and, asymptotically, approaches zero. The authors of this article derived bounds of high-SINR transmission capacity with ALOHA using Poisson point process and CSMA using Mat\'ern point process. Other related works include the analysis of single-hop throughput and capacity with slotted ALOHA, in networks with random and deterministic node placement, and TDMA, in one-dimensional line-networks only, in the article \cite{Haenggi}. 

Most of the works we have discussed so far, analyze the capacity of wireless network by taking into account only the simultaneous single-hop transmissions. However, when source and destination nodes are randomly distributed in the network area, packets are usually transported through multiple relays towards their respective destinations. Therefore, these analyses may not give an accurate assessment of the performance of wireless network using multi-hop communication. Computing the end-to-end capacity of wireless network is an extremely challenging task because of the possible impact of routing protocol on the hop length and spatial distribution of simultaneous transmitters, there have been significant attempts to overcome the above mentioned limitations. In the seminal work of the article \cite{Gupta:Kumar}, the authors used a spatial and temporal scheme for scheduling transmissions in the network and derived the scaling laws on the capacity when the number of nodes increases and approaches infinity. Gupta \& Kumar also introduced the concept of {\em transport capacity} which measures the sum of the end-to-end throughput of the network multiplied by the end-to-end distance and showed that if nodes are optimally located and traffic patterns are also optimally assigned, the network transport capacity scales as $O(\sqrt{n})$ or, in other words, bit-rate per source destination pair in multi-hop networks is upper bounded by $O(1/\sqrt{n})$. However, if nodes are randomly distributed, they showed that the throughput obtainable by each node for a randomly chosen destination drops to $O(1/\sqrt{n\log n})$ where the $\log n$ factor comes from the fact that the network must be connected or, in other words, at least one route or path must exist between all pairs of nodes {\em with high probability} {\em(w.h.p.)} which approaches one as $n$ approaches infinity. The articles \cite{scaling,scaling3} have shown that, even in case of randomly distributed nodes, the throughput capacity in wireless network can be increased to an upper bound of $O(\sqrt{n})$ and the $\log n$ factor can be overcome by constructing a complex system of {\em highways} of nodes to transport data and using different transmission ranges for different transmissions. In other related works, the article \cite{SR-ALOHA} gave a detailed analysis on the optimal probability of transmission for ALOHA which optimizes the spatial density of progress, {\it i.e.}, the product of simultaneously successful transmissions per unit of space by the average range of each transmission. Instead of considering the average range of transmission, the authors of the article \cite{weber08} proposed a variant where each transmitter selects its longest feasible edge and then identifies a packet in its buffer whose next hop is the associated receiver. Similarly, the article \cite{Weber4} evaluated the random transport capacity of wireless network without taking into account any particular routing protocol and under the strong assumptions that simultaneous interferers form a Poisson point process and relays are equally spaced on the straight line between source and destination. 

An important aspect of multi-hop communication is the transmission range of a transmitter as it has an impact on the hop-length as well as the probability of successful transmission: both factors have critical impact on the capacity of multi-hop wireless network. There have been many studies on the transmission range under various constraints. For example, the article \cite{deng-han-varshney2007} studied the transmission range that achieves the most economical energy usage in an ad hoc network of uniformly distributed nodes. The works \cite{santi-blough2003,santi-2005} analyzed the critical transmission range required for connectivity in stationary and mobile wireless ad hoc networks. The article \cite{takagi-kleinrock1984} determined the optimal transmission range of slotted ALOHA and CSMA schemes which maximizes the expected one-hop progress of a packet while considering the tradeoff between the probability of successful transmission and the progress of the packet. The authors assumed that all nodes use same transmit power and the range of transmission is controlled by tuning the density of simultaneous transmitters via parameters like probability of transmission (in case of ALOHA) and carrier sense threshold (in case of CSMA). The article \cite{Zorzi2} determined the normalized optimum transmission range under the assumption that simultaneous interferers form a Poisson point process. In an important analysis in the article \cite{gomez-cambell2004}, the authors studied the impact of variable transmission range, by controlling the transmit power, on network connectivity, capacity and energy efficiency. They showed that routing in networks where transmitters have same transmission range can only achieve half of the traffic carrying capacity of networks where transmitters have variable transmission range. 

\section{Models and Assumptions}
\label{range_sec:model}

In this section, we will discuss the models that we will use throughout this article.

\subsection{Timing Model}

MAC schemes in wireless networks can be divided into two categories: continuous time access and slotted access. In order to simplify our analysis and cope with the difficulties introduced by the continuous time medium access, our main focus in this work will be on the slotted medium access although many of our results can be applied to continuous time medium access as well. 

In slotted medium access, time is divided into slots of equal duration and we consider an ideal timing model, operating with the following assumptions:
\begin{itemize}
\item[--] nodes are synchronized,
\item[--] transmissions begin at the start of the slot,
\item[--] packets are of equal length and 
\item[--] packet transmission time is equal to one slot.
\end{itemize}
Note that, in our timing model, we have assumed that the propagation delays, {\it i.e.}, the time it takes for the signals to travel from the transmitters to the receivers are negligible. 

\subsection{Network Model}

We will consider a wireless network where nodes are uniformly distributed with density $\nu$ in a two-dimensional $(2D)$ area ${\cal A}$ with center at origin $(0,0)$. We denote the set of all nodes in the network by ${\cal N}$ and $n$ is the number of elements in this set. 

Note that, in our theoretical analysis, the set ${\cal N}$ is infinite but distributed with a uniform density $\nu$ in an infinite plane $\mathbb{R}^2$, {\it i.e.}, ${\cal A}=\mathbb{R}^2$. In slotted medium access, at any given slot, simultaneous transmitters in the network are distributed like a set of points 
$$
{\cal S}=\{{\bf z}_1,{\bf z}_2,\ldots,{\bf z}_k,\ldots\}~,
$$ 
where \mbox{${\bf z}_i\in {\cal A}$} is the location of transmitter $i$ and \mbox{${\bf z}_i=(x_i,y_i)$}. The spatial distribution of simultaneous transmitters, {\it i.e.}, the set ${\cal S}$ depends on the MAC scheme employed by nodes. Therefore, we will not adopt any {\em universal} model for the locations of simultaneous transmitters and only assume that, in all slots, the set ${\cal S}$ has a spatially homogeneous density equal to $\lambda$. 

\subsection{Propagation Model}
\label{schemes_sec:pro_model}

The channel gain between an arbitrary transmitter $i$ and a receiver located at point \mbox{${\bf z}\in {\cal A}$} is denoted by $\gamma_i({\bf z})$, such that the received power at the receiver is $P_i\gamma_i({\bf z})$. 

We have the following expression for the channel gain $\gamma_i({\bf z})$
\begin{equation}
\gamma_i({\bf z})=\frac{F_i({\bf z})}{\| {\bf z}-{\bf z}_i\|^{\alpha}}~,
\label{schemes_eq:channel_gain}
\end{equation} 
where $F_i({\bf z})$ is the {\em fading factor} which is a random variable that represents the variation in channel gain between transmitter $i$ and the receiver located at point ${\bf z}$, \mbox{$\|{\bf z}-{\bf z}_i\|$} is the Euclidean distance between the transmitter and the receiver and \mbox{$\alpha>2$} is the attenuation coefficient.  

The variation in channel gain $F_i({\bf z})$ has two components:
\begin{itemize}
\item[--] {\em fast} fading and 
\item[--] {\em slow} fading or shadowing effects. 
\end{itemize} 
We assume that the value of $F_i({\bf z})$ remains constant over the duration of a slot. However, it varies from slot to slot and, for the signal from transmitter $i$ to the receiver located at point ${\bf z}$, there will be correlation between these varying values because of the shadowing effects and, in this case, the mean value of $F_i({\bf z})$ will be less than one. In this article, we will ignore the shadowing effects and we will consider that $F_i({\bf z})$ is an {\em independent and identically distributed} ({\em i.i.d.}) random variable which is independent of the position of the transmitter and the receiver. We will also assume that it is independent of the distance between them which is a simplification. Therefore, we can drop the subscript and the argument {\bf z} and denote \mbox{$F\equiv F_i({\bf z})$}. For the probability distribution of $F$, we will consider the following two cases:

\begin{enumerate}
\item {\em no fading:} $F$ is constant and equal to one,
\item {\em Rayleigh fading:} $F$ is an {\em i.i.d.} random variable with an exponential distribution of mean equal to one. We will discuss this case in detail in \S \ref{sec:impact_fading}.
\end{enumerate}

\subsection{Transmission Model}
\label{schemes_sec:tx_model}

We consider that all nodes are equipped with only one antenna and the transmission from transmitter $i$ to a receiver located at point ${\bf z}$ is successful if and only if the following condition is satisfied
\begin{equation}
\frac{P_i\gamma_i({\bf z})}{N_0+\sum\limits_{j\neq i,{\bf z}_j\in{\cal S}}P_j\gamma_j({\bf z})}\geq \beta~,
\label{schemes_eq:sinr_condition}
\end{equation}
where $N_0$ is the background noise (ambient/thermal) power which is assumed equal to zero and $\beta$ is the minimum SIR threshold required for successfully decoding the packet. 

We consider that all transmitters transmit at the same unit nominal transmit power, {\it i.e.}, 
$$
P_i=1~,
$$ 
for \mbox{$i=0,1,2,\ldots,k,\ldots$}.

Note that, here we have ignored the inter-symbol interference (ISI) caused by echoes which may allow the above condition to be satisfied but the packet may still be non-decodable at the receiver. 

The ISI can be mitigated by using {\em guard periods} and/or adaptive equalization at the receivers however, this is beyond the scope of this article and we will ignore its impact on whether the transmission can be received successfully or not.

Therefore, the SIR of transmitter $i$ at any point ${\bf z}$ is given by
\begin{equation}
S_i({\bf z})=\frac{\gamma_i({\bf z})}{\sum\limits_{j\neq i,{\bf z}_j\in{\cal S}}\gamma_j({\bf z})}~,
\label{schemes_eq:sinr}
\end{equation} 
where $\gamma_i({\bf z})$ is given by \eqref{schemes_eq:channel_gain}.

\subsection{Node Model}
\label{schemes_sec:node_model}

We consider that each node has a buffer which contains packets and a node transmits a packet from its buffer when it is allowed by its MAC layer to transmit on the channel. However, if no packet is available in the buffer, the node transmits a HELLO packet or a {\em junk} packet of length such that its transmission time is equal to the duration of a slot.

Note that, in case of multi-hop networks, each node generates packets in its buffer which are to be transmitted to a randomly selected destination node and packets may have to be routed through multiple relays to reach their destinations. Therefore, the buffers may also carry the routed packets. 

\section{Medium Access Control Schemes}
\label{range_sec:mac_model}

Within slotted medium access category, we will only distinguish slotted ALOHA and grid pattern based MAC schemes in this article. In this section, we will discuss the models of these schemes that we will use in our analysis throughout this article.

\subsection{Slotted ALOHA Scheme}
\label{schemes_sec:aloha}

In slotted ALOHA based MAC scheme, nodes do not use any complicated managed transmission scheduling and transmit their packets independently of each other with a certain medium access probability. 

For theoretical analysis, the set of simultaneous transmitters ${\cal S}$ in each slot can be given by a uniform Poisson distribution of mean equal to $\lambda$ transmitters per unit square area as is also the case in articles \cite{SR-ALOHA,Jacquet:2009,Weber2}. 

\subsection{Grid Pattern Based Medium Access Control Schemes}
\label{schemes_sec:grid_pattern}

It can be argued intuitively that a good MAC scheme should optimize the spatial reuse by positioning the simultaneous transmitters in a regular grid pattern. Evaluation of grid pattern based MAC schemes is interesting as they can provide an {\em ideal} upper-bound on the performance, {\it e.g.}, under {\em no fading} channel model. However, their implementation is difficult because of the following reasons:
\begin{itemize}
\item[--] limitations introduced by wave propagation characteristics like shadowing effects and
\item[--]  irregular distribution of nodes (however, simultaneous transmitters can be base stations in cellular networks and their placement can be optimized). 
\end{itemize}
In wireless networks, location aware nodes may be useful and if the node density is high, simultaneous transmitters can be selected in such a way that they closely resemble an ideal grid pattern. However, specification of a distributed MAC scheme that would allow grid pattern based positioning of simultaneous transmitters is beyond the scope of this article. 

\begin{figure*}[!t]
\centering
\psfrag{a}{$\sqrt{3}d$}
\psfrag{b}{$2d$}
\psfrag{c}{$d$}
\psfrag{e}{$k_1d$}
\psfrag{f}{$k_2d$}
\hspace{-1 cm}
\includegraphics[scale=0.5]{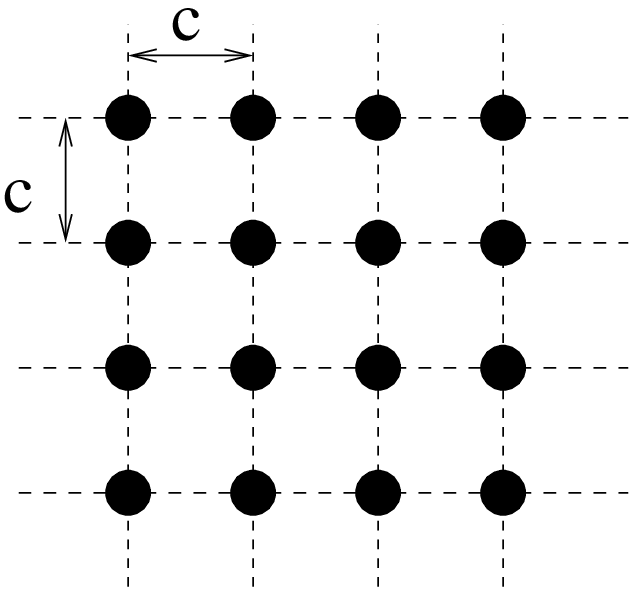}
\hspace{+2 cm}
\includegraphics[scale=0.5]{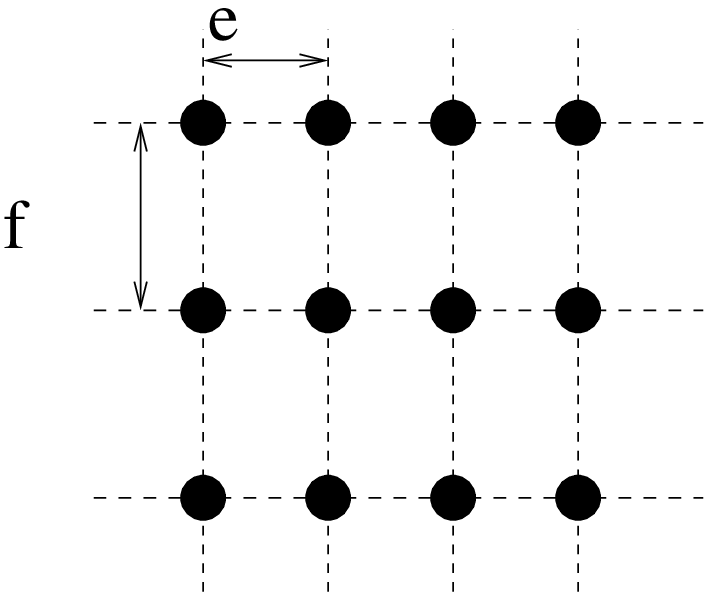}

\hspace{-1 cm}
\includegraphics[scale=0.5]{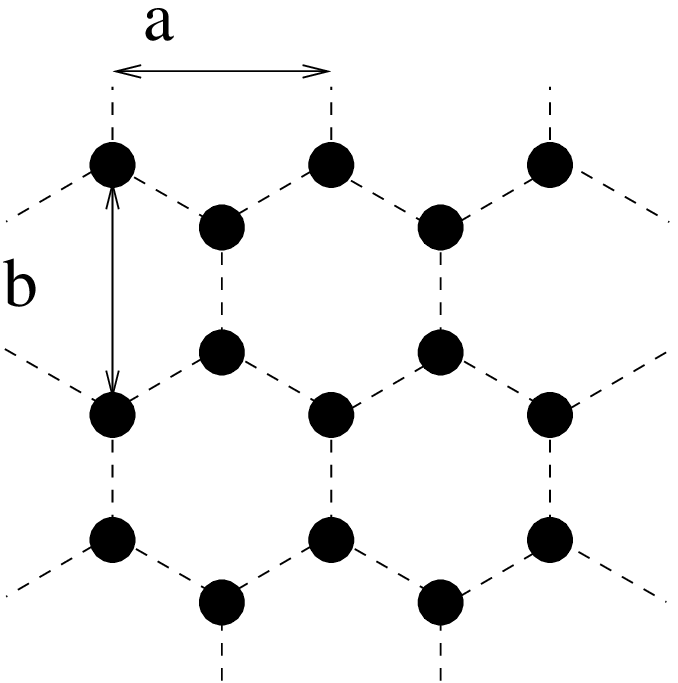}
\hspace{+2 cm}
\includegraphics[scale=0.5]{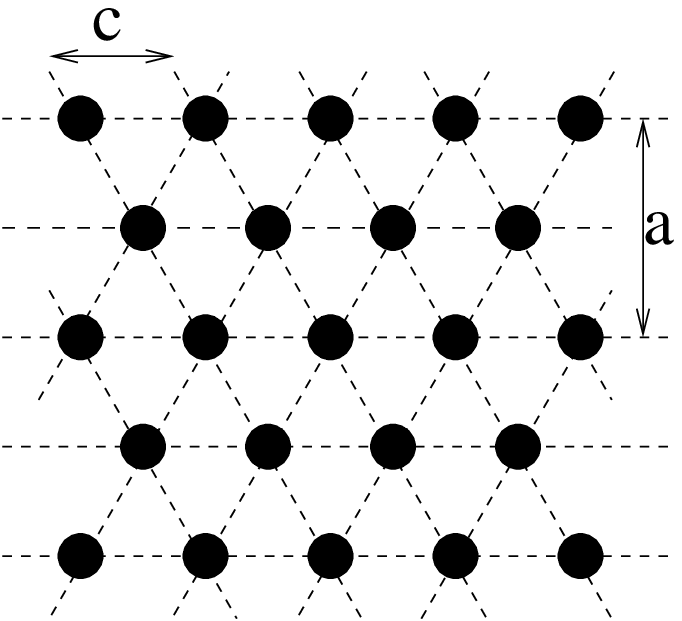}
\caption{Square, Rectangular, Hexagonal and Triangular grid patterns. Note that $\frac{k_1}{k_2}<1$.}
\label{schemes_fig:grid_layouts}
\end{figure*}

We consider the following model for the grid pattern based schemes. The set of nodes ${\cal N}$ is divided into $k$ partitions such that each partition is a grid (or closely resembles a grid) with parameter $d$. The transmission process is a {\em round robin} scheduling on the $k$ partitions, {\it i.e.}, in each slot, a partition is randomly selected and all nodes in that partition transmit simultaneously. Note that the parameter $d$ defines the minimum distance between neighboring transmitters and it can be derived from the hop-distance parameter of a typical TDMA protocol. In our analysis, we have covered grid layouts of square, rectangle, hexagonal and triangle which are shown in Fig. \ref{schemes_fig:grid_layouts}.

Instead of handling the complication of the above described pattern fitting process that converges very slowly, we have devised the following strategy to evaluate grid pattern based schemes. We consider that the set ${\cal S}$ is a set of points arranged in a regular grid pattern. The basic mechanism for the construction of the set of simultaneous transmitters ${\cal S}$ under a grid pattern based scheme, in a network with randomly distributed nodes, is outlined in the following steps. 
\begin{enumerate}
\item Initialize ${\cal M}={\cal N}$ and ${\cal S}=\emptyset$.
\item Randomly select a node $i$ from ${\cal M}$ and add it to the set ${\cal S}$, {\it i.e.},
$$
{\cal S}={\cal S}\cup\{{\bf z}_i\}~.
$$ 
\item Remove node $i$ from the set ${\cal M}$.
\item From the position of node $i$, construct a {\em virtual} grid pattern with parameter $d$. This virtual grid pattern should be one of the patterns we have mentioned above and it is also randomly rotated about the location of node $i$.
\item All nodes which overlap the points of the virtual grid, constructed in the above step, shall be added to the set ${\cal S}$. 
\end{enumerate}
We can assume that the density of nodes in the network is high and the points of the virtual grid always overlap existing nodes. Note that the above steps are repeated in each slot to construct the set of simultaneous transmitters ${\cal S}$ and, in each slot, the virtual grid pattern is the same {\em modulo} a translation and/or rotation. Note that wireless networks of grid topologies are studied in literature, {\it e.g.}, the articles \cite{Liu:Haenggi,Hong:Hua} compared these networks with network consisting of randomly distributed nodes. In contrast to these works, we assume that wireless network consists of randomly distributed nodes and only the simultaneous transmitters form a regular grid pattern.

\section{Performance Analysis}
\label{sec_performance_analysis}

In this section, we will evaluate the performance of the MAC schemes that we discussed in \S \ref{range_sec:mac_model} and we will only use the {\em no fading} channel model. Fading may also influence the performance of these schemes and we will evaluate its impact in detail in \S \ref{sec:impact_fading} where we will use the {\em Rayleigh fading} channel model.

We call the reception area of transmitter $i$ at location \mbox{${\bf z}_i\in {\cal A}$}, the area of the plane ${\cal A}_{i}({\cal S},\beta,\alpha)$ where this transmitter is received with SIR at least equal to $\beta$.

Under {\em no fading} channel model, ${\cal A}_{i}({\cal S},\beta,\alpha)$ can be defined as
\begin{equation}
\mathcal{A}_i({\cal S},\beta,\alpha)=\{\mathbf{z}\,:\, 
\gamma_i({\bf z})\geq\beta\sum_{j\neq i,{\bf z}_j\in{\cal S}}\gamma_j({\bf z})\}~.
\label{local_capacity_eq:rx_area}
\end{equation}
We recall from \S \ref{schemes_sec:pro_model} that the background noise (ambient/thermal) power $N_0$ is assumed negligible and equal to zero and the channel gain $\gamma_i({\bf z})$ under {\em no fading} channel model is given by
$$
\gamma_i({\bf z})=\frac{1}{\|\mathbf{z}-\mathbf{z}_{i}\|^{\alpha}}~.
$$

In the article \cite{Gupta:Kumar}, Gupta \& Kumar showed that under general settings, the efficient radius of transmission is \mbox{$r=r_{\lambda}=\frac{\kappa}{\sqrt{\lambda}}$}, for some \mbox{$\kappa>0$} which depends on the protocol, signal propagation and demodulation, {\it etc}. 

If $\overline{L}$ is the distance between a source and its destination node, the number of relays required to deliver a packet to the destination is equal to \mbox{$\lceil\frac{\overline{L}}{r}\rceil\sim\overline{L}\frac{\sqrt{\lambda}}{\kappa}$}. If all packets were to be moved over an average distance $\overline{L}$, the average number of times each packet would be repeated is given by \mbox{$\frac{\overline{L}}{r p(\lambda,r,\beta,\alpha)}= \overline{L}\frac{\sqrt{\lambda}}{\kappa p(\lambda,r,\beta,\alpha)}$} where $p(\lambda,r,\beta,\alpha)$ is the probability of successfully transmitting a packet at distance $r$ under given values of $\beta$ and $\alpha$. 

Therefore, the density of useful network capacity would be 
$$
\lambda\frac{r p(\lambda,r,\beta,\alpha)}{\overline{L}}=\sqrt{\lambda}\frac{\kappa p(\lambda,r,\beta,\alpha)}{\overline{L}}
$$ 
packets per slot.

In this article, we will evaluate the following parameters in a multi-hop wireless network:
\begin{itemize}
\item[--] the minimized number of transmissions required to transport a packet through multiple relays before delivery to its final destination located at unit distance from the source, {\it i.e.}, $\frac{1}{r p(\lambda,r,\beta,\alpha)}$ with $\lambda=1$ and
\item[--] the normalized optimum transmission range ${r_1}$ that minimizes the above quantity. 
\end{itemize}

Note that \mbox{$r_1=\sqrt{\lambda}r_\lambda$} as it is invariant for any homothetic transformation of the set of simultaneous transmitters ${\cal S}$. We will measure these parameters for slotted ALOHA and compare its results with various grid pattern based schemes. 

\subsection{Slotted ALOHA Scheme}
\label{range_sec:aloha}

We discussed in \S \ref{schemes_sec:aloha} that in case of slotted ALOHA scheme, the set of simultaneous transmitters in each slot can be given by a uniform Poisson distribution of mean equal to $\lambda$ transmitters per unit square area. 

Using \eqref{schemes_eq:sinr_condition} with \mbox{$N_0=0$} and $P_i=1$ (for \mbox{$i=0,1,\ldots,k,\ldots$}), we can write
$$
\|{\bf z}-{\bf z}_i\|^{-\alpha}\geq \beta\sum_{j\neq i,{\bf z}_j\in {\cal S}}\|{\bf z}-{\bf z}_j\|^{-\alpha}~,
\label{local_capacity_eq:snr_relation_at_z}
$$
or 
$$
{\cal W}({\bf z},\{{\bf z}_i\})\geq \beta {\cal W}({\bf z},{\cal S}-\{{\bf z}_i\})~,
$$ 
where 
$$
{\cal W}({\bf z},{\cal S})=\sum_{{\bf z}_j\in {\cal S}}({\bf z})\|{\bf z}-{\bf z}_j\|^{-\alpha}=\sum_{{\bf z}_j\in {\cal S}}\|{\bf z}-{\bf z}_j\|^{-\alpha}~.
$$ 

\subsubsection{Evaluation of Transmission Range}

\begin{figure}[!t]
\centering
\includegraphics[scale=0.65]{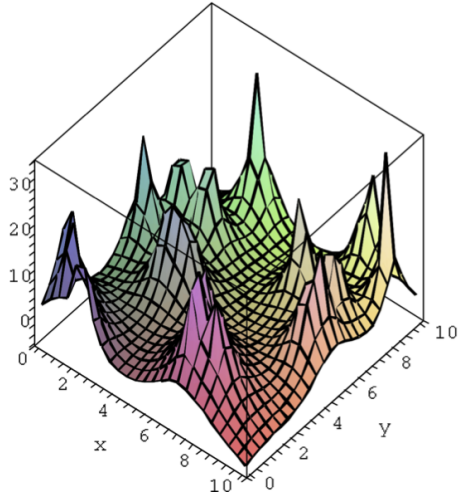}
\caption{Signal Levels (in dB's) for a random network with attenuation coefficient $\alpha=2.5$.
\label{local_capacity_fig:signal_levels}}
\end{figure}

\begin{figure}[!t]
\centering
\includegraphics[clip,scale=0.65]{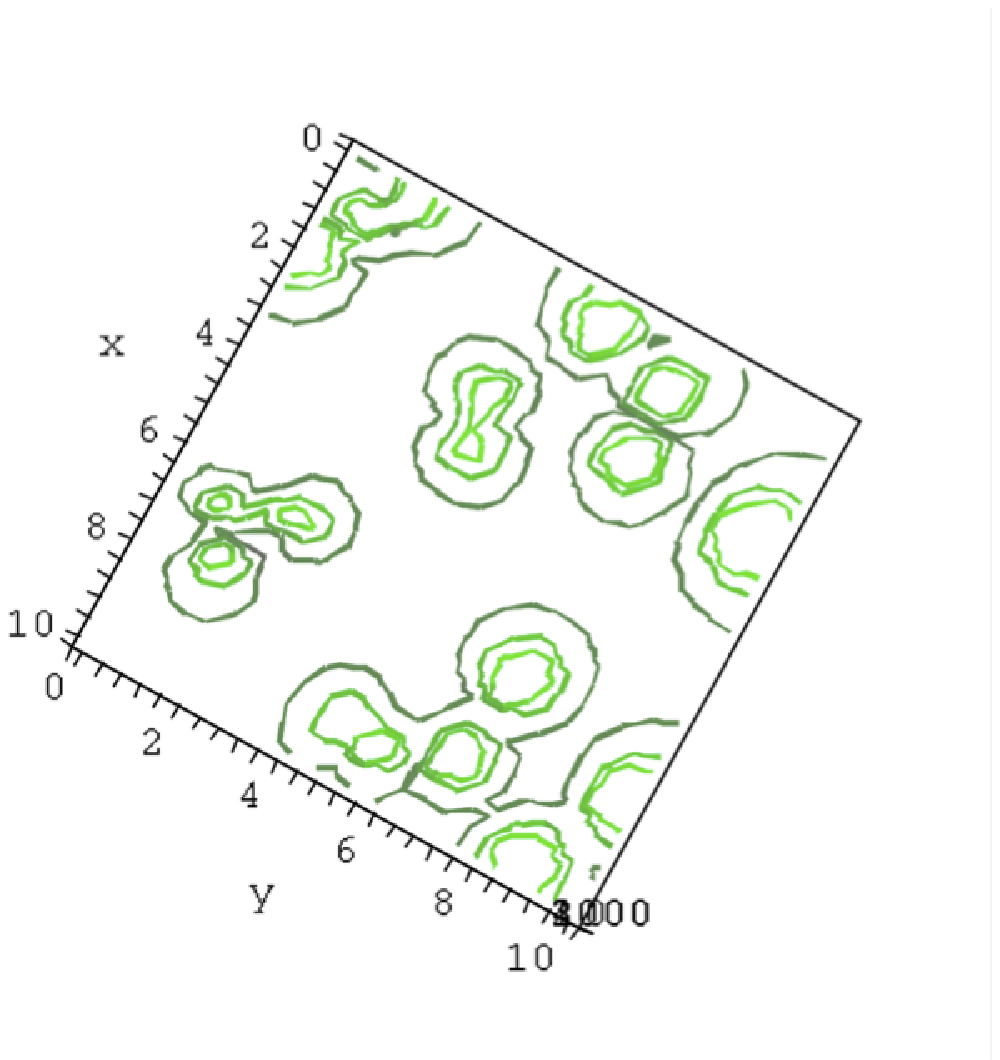}
\caption{Distribution of reception areas for various value of SIR. $\beta=1,4,10$ for situation of Fig. \ref{local_capacity_fig:signal_levels}.
\label{local_capacity_fig:reception_area}}
\end{figure}

Figure \ref{local_capacity_fig:signal_levels} shows the function ${\cal W}({\bf z},{\cal S})$ for ${\bf z}$ varying in the plane with ${\cal S}$, an arbitrary set of Poisson distributed transmitters. Figure \ref{local_capacity_fig:signal_levels} uses \mbox{$\alpha=2.5$}. It is clear that closer the receiver is to the transmitter, larger is the SIR. For each value of $\beta$ we can draw an area, around each transmitter, where its signal can be received with SIR greater or equal to $\beta$. Figure \ref{local_capacity_fig:reception_area} shows reception areas for the same set ${\cal S}$, as in Fig. \ref{local_capacity_fig:signal_levels}, for various values of $\beta$. As can be seen, the reception areas do not overlap for \mbox{$\beta>1$} since there is only one dominant signal. For each value of $\beta$ we can draw, around each transmitter, the area where its signal is received with SIR greater or equal to $\beta$. Our aim is to find the probability of successfully transmitting a packet $p(\lambda,r,\beta,\alpha)$ at distance $r$ and how it is a function of $\lambda$, $\beta$ and $\alpha$.

\begin{theorem}
In case of no fading channel model, the probability of successfully transmitting a packet to a receiver at distance $r$ with slotted ALOHA scheme is
\begin{align}
p(\lambda,r,\beta,\alpha)&=\Pr\left({\cal W}(\lambda)<\frac{r^{-\alpha}}{\beta}\right)\notag\\
&=\underset{n\geq0}{\sum}\frac{(-C\lambda)^{n}}{n!}\frac{\sin(\pi n\gamma)}{\pi}\Gamma(n\gamma){\left(\frac{r^{-\alpha}}{\beta}\right)}^{-n\gamma}~,
\label{range_eq:prob_aloha}
\end{align}
where the expression of $C$ in case of $2D$ map is 
$$
C=\pi\Gamma(1-\gamma)~,
$$ 
\mbox{$\gamma=\frac{2}{\alpha}$} and $\Gamma(.)$ is the Gamma function.
\end{theorem}

\begin{proof}
${\cal W}({\bf z},{\cal S})$ depends on ${\cal S}$ and hence is also a random variable. The random variable ${\cal W}({\bf z},{\cal S})$ has a distribution which is invariant by translation and therefore does not depend on ${\bf z}$. Therefore, we denote ${\cal W}(\lambda)\equiv{\cal W}({\bf z},\lambda)$. Let $w({\cal S})$ be its density function. The set ${\cal S}$ is given by a $2D$ Poisson process process with intensity $\lambda$ transmitters per slot per unit square area and Laplace transform of $w({\cal S})$, denoted by $\tilde{w}(\theta,\lambda)$, can be computed exactly, {\it i.e.}, 
$$
\tilde{w}(\theta,\lambda)=\E[\exp(-\theta {\cal W}(\lambda))]~.
$$ 

We split the network area into small sub-rectangles of size \mbox{$dx\times dy$}. As the set ${\cal S}$ is distributed according to the Poisson point process, the contributions of all sub-areas are independent and the Laplace transform is equal to the product of the Laplace transforms of the contributions of all sub-areas. The contribution of a subarea at distance $r$ from ${\bf z}$ is
$$
\exp(-\lambda dx dy)+(1-\exp(-\lambda dx dy))e^{-\theta r^{-\alpha}}~.
$$
This expression means that with probability $\exp(-\lambda dx dy)$, there is no transmitter in the sub-rectangle and with probability $(1-\exp(-\lambda dx dy))$, there is a transmitter. Note that we consider only one transmitter as $dx$ and $dy$ are infinitesimally small and the contribution of this transmitter to the Laplace transform is equal to $e^{-\theta r^{-\alpha}}$. 

Using the first order estimate,
\begin{align}
\exp(-\lambda dx dy)&=1-\lambda dx dy + O(dx^2 dy^2)\notag\\
\intertext{we obtain}
\log \tilde{w}(\theta,\lambda)&=\int\int \lambda (e^{-\theta r^{-\alpha}}-1)dx dy\notag\\
&=2\pi\lambda \int (e^{-\theta r^{-\alpha}}-1) r dr.\notag\\
\intertext{Changing variable $r^{-\alpha}=x$ which leads to $\frac{dx}{x}=-\alpha\frac{dr}{r}$, we get}
\log \tilde{w}(\theta,\lambda)&=-\frac{1}{\alpha}2\pi \lambda \int (e^{-x}-1) \frac{x^{-\frac{2}{\alpha}-1}}{\theta^{-\frac{2}{\alpha}-1}}\frac{dx}{\theta}\notag\\
&=-\pi \lambda \Gamma\left(1-\frac{2}{\alpha}\right)\theta^{\frac{2}{\alpha}}\notag~.
\intertext{Therefore, we get}
\tilde{w}(\theta,\lambda)&=\exp(-\pi \lambda \Gamma(1-\gamma)\theta^{\gamma})~.
\label{local_capacity_eq:laplace}
\end{align}
Note that, in all cases, $\tilde{w}(\theta,\lambda)$ is of the form $\exp(-\lambda C\theta^{\gamma})$ and $C=\pi\Gamma(1-\gamma)$. 

From \eqref{local_capacity_eq:laplace} and by applying the inverse Laplace transformation, we get
$$
\Pr({\cal W}(\lambda)<x)=\frac{1}{2i\pi}\underset{-i\infty}{\overset{+i\infty}{\int}}\frac{\tilde{w}(\theta,\lambda)}{\theta}e^{\theta x}d\theta~,
$$
Expanding $\tilde{w}(\theta,\lambda)=\underset{n\geq0}{\sum}\frac{(-C\lambda)^{n}}{n!}\theta^{n\gamma}$, we get
$$
\Pr({\cal W}(\lambda)<x)=\frac{1}{2i\pi}\underset{n\geq0}{\sum}\frac{(-C\lambda)^{n}}{n!}\underset{-i\infty}{\overset{+i\infty}{\int}}\theta^{n\gamma-1}e^{\theta x}d\theta~,
$$
By bending the integration path towards negative axis
\begin{equation*}
\frac{1}{2i\pi}\underset{-i\infty}{\overset{+i\infty}{\int}}\theta^{n\gamma-1}e^{\theta x}d\theta=\frac{e^{i\pi n\gamma}-e^{-i\pi n\gamma}}{2i\pi}\underset{0}{\overset{\infty}{\int}}\theta^{n\gamma-1}e^{-\theta x}d\theta=\frac{\sin(\pi n\gamma)}{\pi}\Gamma(n\gamma)x^{-n\gamma}~,
 \end{equation*}
we get,
\begin{equation}
\Pr({\cal W}(\lambda)<x)=\underset{n\geq0}{\sum}\frac{(-C\lambda)^{n}}{n!}\frac{\sin(\pi n\gamma)}{\pi}\Gamma(n\gamma)x^{-n\gamma}\label{local_capacity_eq:signal_pd}~,
\end{equation}
with the convention that $\frac{\sin(\pi n \gamma)}{\pi}\Gamma(n\gamma)=1$ for $n=0$.
\end{proof}

Figure \ref{range_fig:dist_aloha} shows the plot of $p(\lambda,r,\beta,\alpha)$ versus $r$ for $\lambda=1$, $\beta=1.0$ and various values of $\alpha$. 

\begin{figure}[!t]
\centering
\psfrag{r}{$r$}
\includegraphics[scale=0.9]{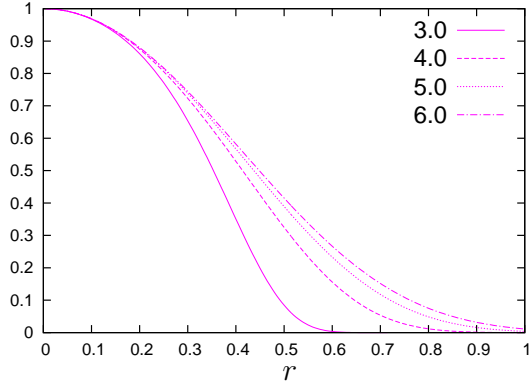}
\caption{$p(1,r,\beta,\alpha)$ versus $r$ of slotted ALOHA. $\lambda=1$, $\beta=1.0$ and $\alpha$ is varied from $3.0$ to $6.0$.
\label{range_fig:dist_aloha}}
\end{figure}

Here, we will use simple algebra to show the homothetic invariance of $r_1$ which has also been proved in the article \cite{adjih04}.

\begin{theorem}
$r_1=\sqrt{\lambda}r_{\lambda}$, where $r_{\lambda}$ is the radius under the given density $\lambda$ of the set of simultaneous transmitters ${\cal S}$ such that 
$$
\Pr\left({\cal W}(\lambda)<\frac{{r_{\lambda}}^{-\alpha}}{\beta}\right)=\int_0^{\frac{{r_{\lambda}}^{-\alpha}}{\beta}}w(x,\lambda)dx=p_0~,
$$
where $p_0$ is a constant under given values of $\alpha$ and $\beta$. 
\end{theorem}
\begin{proof}
The signal level at any point in the plane with Poisson distributed transmitters is a random variable and the Laplace transform of its probability density is given by \eqref{local_capacity_eq:laplace}.
Using the reverse Laplace transform we have 
$$
w(x,\lambda)=\frac{1}{2i\pi}\int_{-i\infty}^{+i\infty}\tilde{w}(\theta,\lambda)e^{\theta x}d\theta~.
$$
Inserting the expression from \eqref{local_capacity_eq:laplace} in the above equation and commuting the integral signs because 
$$
\int_0^{\frac{{r_{\lambda}}^{-\alpha}}{\beta}}e^{\theta x}dx=\frac{e^{\theta \frac{{r_{\lambda}}^{-\alpha}}{\beta}}-1}{\theta}~,
$$
yields
$$
\frac{1}{2i\pi}\int_{-i\infty}^{+i\infty}\frac{e^{\theta \frac{{r_{\lambda}}^{-\alpha}}{\beta}}-1}{\theta}\tilde{w}(\theta,\lambda)d\theta=p_0~.
$$
The change of variable $\lambda^{\frac{\alpha}{2}}\theta=\theta'$, makes $\lambda$ disappear from the $\tilde{w}(\theta,\lambda)$ expression and we get
$$
\frac{1}{2i\pi}\int_{-i\infty}^{+i\infty}\frac{e^{\theta' \frac{{r_{\lambda}\sqrt{\lambda}}^{-\alpha}}{\beta}}-1}{\theta'}\tilde{w}(\lambda^{\frac{-\alpha}{2}}\theta',1)d\theta'=p_0~.
$$
Since $\tilde{w}(\lambda^{\frac{-\alpha}{2}}\theta',1)$ is independent of $\lambda$ and $r_{\lambda}$ is multiplied by $\sqrt{\lambda}$, we get that $r_{\lambda}$ is simply proportional to $1/\sqrt{\lambda}$, {\it i.e.}, $r_{\lambda}=r_1/\sqrt{\lambda}$.
\end{proof}

Our objective is to derive the optimum value of $r_1$ which maximizes the function  $rp(\lambda,r,\beta,\alpha)$ for $\lambda=1.0$. 

\subsection{Grid Pattern Based Medium Access Control Schemes}
\label{range_sec:grids}

In this section also, we will use the {\em no fading} channel model only. 

Under {\em no fading} channel model, when activated by a grid pattern based scheme, the reception areas of all transmitters remain the same, in every slot, {\em modulo} a translation and/or rotation and can be defined deterministically. Therefore, by virtue of a grid pattern, we can consider {\em w.l.o.g.} the transmitter $i$ located at origin, {\it i.e.},
$$
{\bf z}_i={\bf z}_0=(x_0,y_0)=(0,0)~.
$$ 
Let us consider any point ${\bf z}$ in the plane such that
$$
r_{\bf z}=\|{\bf z}-{\bf z}_0\|~,
$$
and the probability to receive a signal at point ${\bf z}$, which is at distance $r_{\bf z}$ from transmitter $i$ at origin, with SIR at least equal to $\beta$ is
\begin{align*}
p(\lambda, r_{\bf z}, \beta, \alpha)&=1,\;\;\; \text{if}\; {\bf z}\in \mathcal{A}_i({\cal S},\beta,\alpha)\\
p(\lambda, r_{\bf z}, \beta, \alpha)&=0,\;\;\; \text{otherwise}~.
\end{align*}

Note that the grid pattern may be rotated from slot to slot and the reception areas of all transmitters may also be rotated though, under any given orientation, they are defined deterministically for all slots. We defined the probability of successful reception of transmitter $i$ at a receiver at point ${\bf z}\in \mathcal{A}_i({\cal S},\beta,\alpha)$ as equal to one because we do not average it over all rotation orientations during which the transmitter $i$ is active. The reason is that, in our model of packet relaying scheme, a transmitter will transmit a packet to a receiver which is located within its reception area and the probability of successfully transmitting this packet is one under {\em no fading}. We will describe this model in detail in the following discussion. However, in case of {\em Rayleigh fading}, we cannot define the reception area deterministically and, therefore, we will have to introduce a measure for the probability of successful reception which takes into account the variations in received signal levels. The impact of fading on the optimum transmission range of grid pattern based schemes will be discussed in detail in \S \ref{sec:impact_fading_grids}. 

Therefore, under {\em no fading} channel model, we can define the optimum transmission range of a transmitter, when nodes employ grid pattern based schemes, as equal to its maximum transmission range, {\it i.e.}, 
$$
r_{\lambda}=\{\max\|{\bf z}-{\bf z}_0\| : \gamma_i({\bf z})\geq \beta\sum\limits_{j\neq i,j\in{\cal S}}\gamma_j({\bf z})\}~,
$$ 
which is the maximum distance within the reception area of a transmitter at which it can transmit successfully. Therefore, the minimum average number of retransmissions required to deliver a packet to its final destination at mean distance $\overline{L}$ becomes equal to $\overline{L}/r_{\lambda}$. 

As we discussed earlier, the purpose of our analysis of grid pattern based schemes is to establish bounds on the capacity in wireless networks. In this section also, we are interested in establishing an upper bound on the normalized transmission range and a lower bound on the number of retransmissions required to deliver a packet in multi-hop networks. That is why, we are interested in determining the maximum transmission range of a transmitter with various grid pattern based schemes. As we have introduced the notion of maximum transmission range, we will first briefly discuss the progression of a packet, towards its destination, in multi-hop wireless networks with grid pattern based schemes. Later we will present the numerical method we have used to determine the maximum transmission range in wireless networks with grid pattern based schemes.

\begin{figure}[!t]
\centering
\psfrag{s}{$s$}
\psfrag{a}{$r_1$}
\psfrag{b}{$d_1$}
\psfrag{c}{$r_2$}
\psfrag{d}{$d_2$}
\psfrag{e}{$r_3$}
\psfrag{f}{$d_3$}
\includegraphics[scale=0.8]{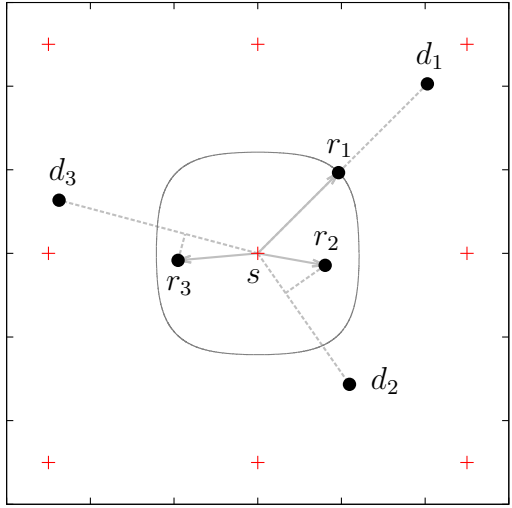}
\caption{An example: node $s$ carries packets for destination nodes $d_1$, $d_2$ and $d_3$ and transmits the packet for the destination $d_1$ towards the relay $r_1$. Nodes in the network employ square grid pattern based scheme.}
\label{range_fig:progress}
\end{figure}

\begin{figure*}[!t]
\centering
\psfrag{s}{$s$}
\psfrag{d}{$d$}
\subfloat[$1^{st}$ hop progress.]
{
	\hspace{-.35cm}
\includegraphics[scale=0.8]{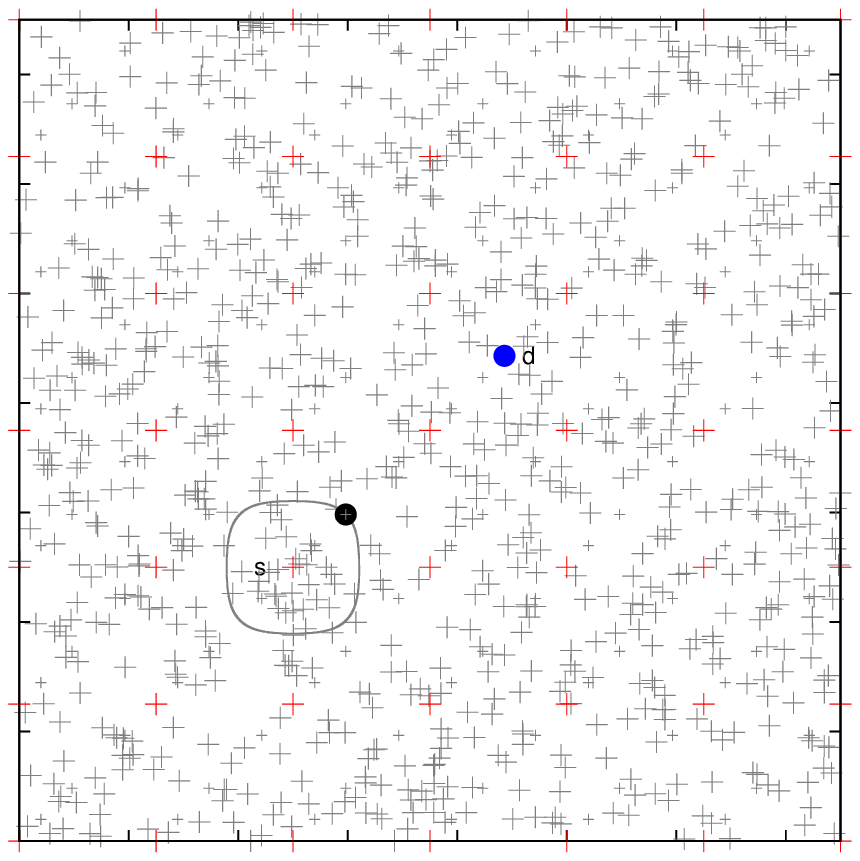}
}
\subfloat[$2^{nd}$ hop progress.]
{
\includegraphics[scale=0.8]{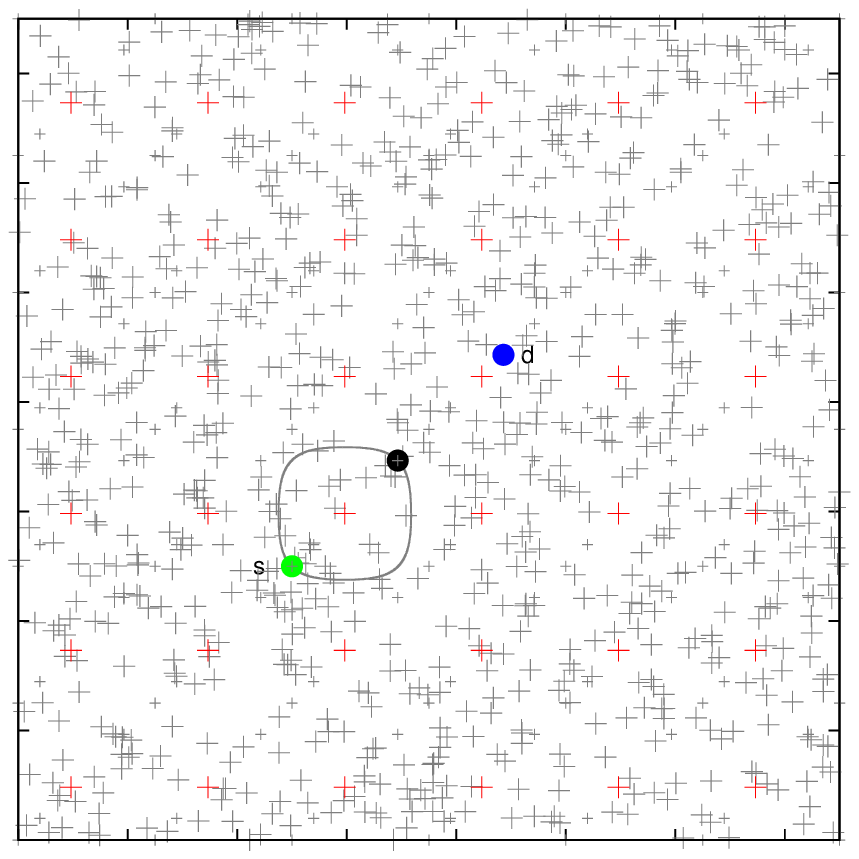}
}

\subfloat[$3^{rd}$ hop progress.]
{
	\hspace{-.35cm}
\includegraphics[scale=0.8]{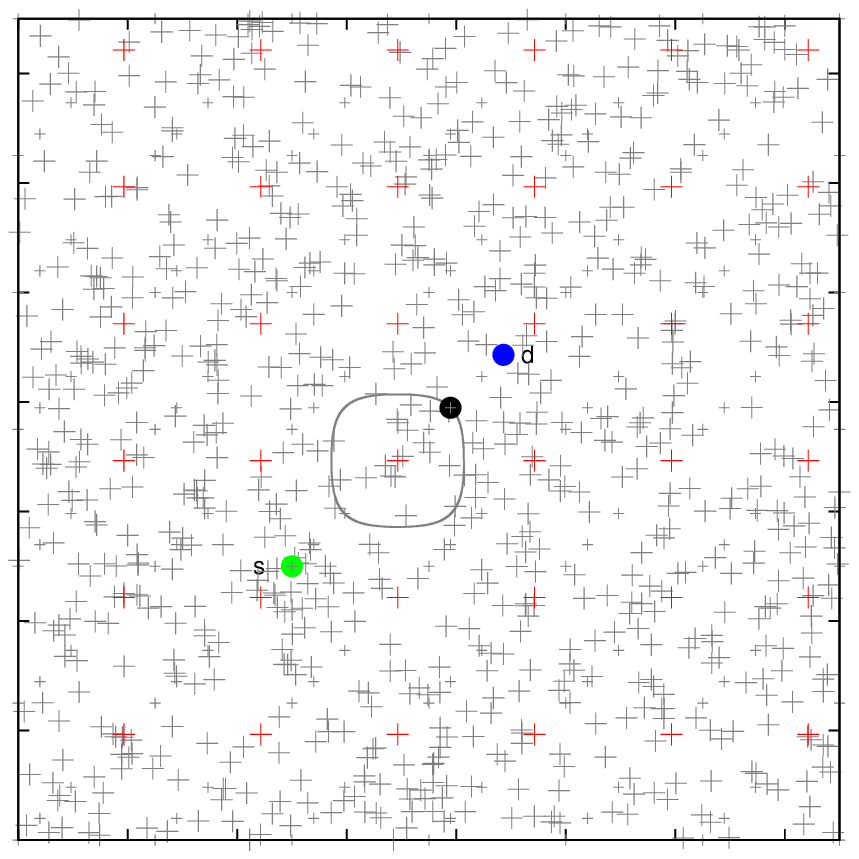}
}
\subfloat[$4^{th}$ hop progress (packet delivered).]
{
\includegraphics[scale=0.8]{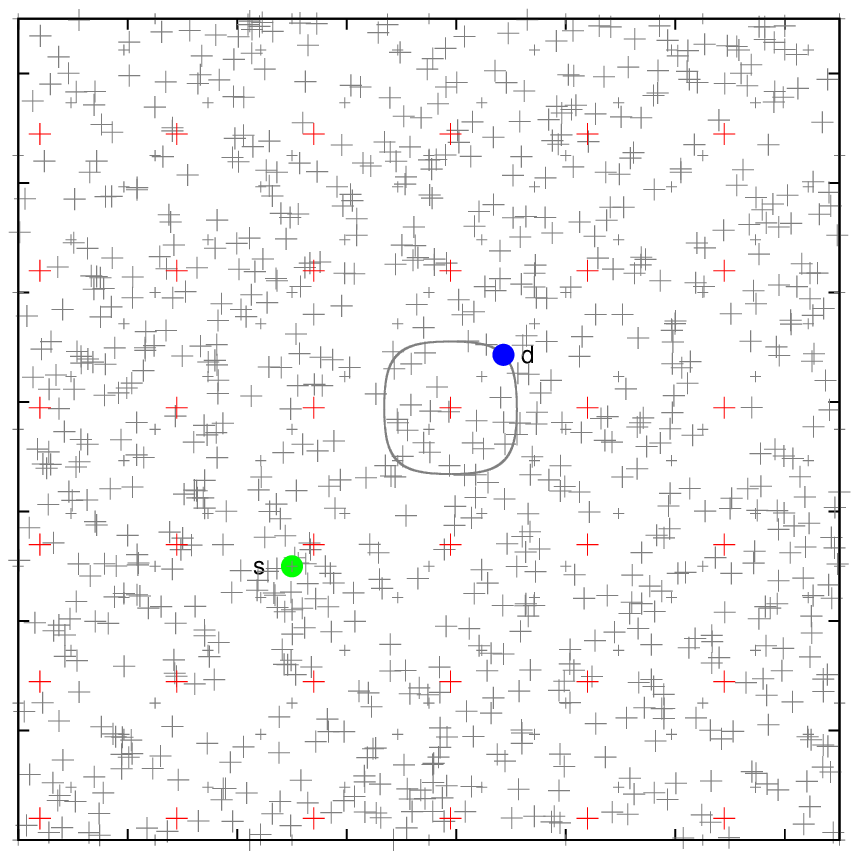}
}
\caption{An example: progress of a packet from the source node $s$ towards the destination node $d$ with square grid pattern based scheme. \label{range_fig:protocol_model}}
\end{figure*}

We define the {\em progress} of a packet as the distance between the transmitter and the receiver relay projected onto the line joining the transmitter and the final destination node of the transmitted packet. A node $j$ is said to be in transmitter $i$'s forward direction if the progress is non-negative when transmitter $i$ transmits packet to the receiver node $j$ successfully. Otherwise, node $j$ is said to be in the backward direction of transmitter $i$.

We assume that a transmitter in the network is also aware of the following parameters:
\begin{itemize}
\item[--] its own cartesian coordinates,
\item[--] the cartesian coordinates of all destination nodes for which it is carrying a packet (this information can be contained in the packet itself) and
\item[--] the cartesian coordinates of all receiver nodes located within its reception area. In order to obtain this information, we may assume that a slot is sub-divided into management and data sub-slots and this information is collected during the management sub-slot and the packet is later transmitted during the data sub-slot. However, this specification is beyond the scope of this article and we only assume that this information is available at all transmitters.
\end{itemize}
Using this information, a transmitter can determine to transmit the packet with the largest forward progress (this condition maybe relaxed for the packets that can be delivered to their destination nodes immediately). This ensures that the number of hops required to reach the destination are minimized. Note that, we assumed that a transmitter, when activated for transmission by the grid pattern based scheme, always has a packet to transmit. Consider the example in Fig. \ref{range_fig:progress} where nodes in the network employ square grid pattern based scheme. This figure shows a set of simultaneous transmitters where the transmitter $s$ is carrying packets for the destination nodes $d_1$, $d_2$ and $d_3$. Note that, the closed curve bounding the node $s$ is the boundary of its reception area. As observed in the figure, the packet for destination $d_1$ has the largest forward progress and, therefore, node $s$ transmits this packet to the relay node $r_1$.

Figure \ref{range_fig:protocol_model} is the figurative representation of the progress of a packet from node $s$ towards node $d$ in a network with randomly distributed nodes. In this case, nodes employ square grid pattern based scheme. Under this scheme, simultaneous transmitters, in each slot, are selected according to the rules specified in \S \ref{schemes_sec:grid_pattern} and they form a regular square grid pattern which is {\em modulo} a translation and/or rotation. In the figure, the set of simultaneous transmitters are represented by the red color and the inactive nodes (potential receivers) are represented by the gray color. The source, destination and successful relay nodes are represented by green, blue and black colors respectively. When the source node $s$ is activated for transmission by the MAC scheme, it transmits the packet to the next relay. Similarly, Figures \ref{range_fig:protocol_model}(b)-(d) show the progress of the packet from the successive relay nodes towards the final destination of the packet. 

The maximum forward progress of a packet is equal to the maximum transmission range achievable with the MAC scheme employed by nodes. As the density of nodes in the network increases and under the assumption that simultaneous transmitters always have packets to be transmitted in the direction of their maximum transmission range, the forward progress of the packets approach the maximum transmission range and packets travel in straight lines from their source nodes towards their destination nodes.  If $\overline{L}$ is the mean distance between source and its destination node, we can see that the minimum number of hops required to reach the destination node is exactly equal to $\frac{\overline{L}}{r_{\lambda}}$, where $r_{\lambda}$ is the maximum transmission range under the given density of the set ${\cal S}$. In this chapter, we will evaluate the maximum transmission range of a transmitter with various grid pattern based schemes. We will cover the grid patterns of square, rectangular, hexagonal and triangle as shown in Fig. \ref{schemes_fig:grid_layouts}. For grid pattern based schemes, we are interested in the maximum transmission range and it is intuitive to assume that the rectangular grid pattern maybe able to achieve higher transmission range as compared to other grid patterns because transmitters along one dimension are far apart as compared to the other dimension. Therefore, under certain conditions, we maybe able to achieve better bounds on the  capacity in multi-hop wireless networks with rectangular grid pattern based scheme. Our discussion in the following sections will help us in identifying those conditions. Note that, the density of the set of simultaneous transmitters $\lambda$ depends on parameter $d$. However, \mbox{$r_1=\sqrt{\lambda}r_{\lambda}$} is independent of the value of $d$ or, for that matter, $\lambda$ as it is invariant for any homothetic transformation of the set of simultaneous transmitters ${\cal S}$. 

\subsubsection{Evaluation of Transmission Range}
\label{range_sec:tx_range_method}

Because of the correlation between the location of simultaneous transmitters in case of grid pattern based MAC schemes, it does not seem possible to develop a tractable analytical model and derive the closed-form expression for the distribution of signal levels and the maximum transmission range. There have been attempts to derive bounds on the interference in a wireless network consisting of nodes placed in a regular pattern. For instance, the article \cite{haenggi09} derived the lower bound in such networks by computing the signal level of interference at the center of the network. In contrast, we will propose a method for accurate demarcation of the reception area of a transmitter under {\em no fading} and we will use this method to compute the maximum transmission range in wireless networks with grid pattern based schemes. 

In grid pattern based schemes, the set of simultaneous transmitters, {\it i.e.} the set ${\cal S}$, is arranged in a grid pattern. We will cover grids of square, rectangular, hexagonal and triangle patterns and, by consequence of the grid pattern, all transmitters have the same maximum transmission range {\em modulo} a translation and/or rotation. We will use the method of Lagrange multipliers, described in Chapter $17$ of the book \cite{lagrange}, to find the maximum transmission range of transmitter $i$ located at the origin while fulfilling the constraint of SIR threshold. 

\begin{figure}[!t]
\centering
\psfrag{a}{${\bf z}_i$}
\psfrag{b}{${\bf z}$}
\psfrag{c}{$d{\bf z}=J\frac{\nabla S_{i}({\bf z})}{|\nabla S_{i}({\bf z})|}\delta t$}
\psfrag{d}{${\cal C}_i({\cal S},\beta,\alpha)$}
\psfrag{e}{$A_i({\cal S},\beta,\alpha)$}
\includegraphics[scale=0.65]{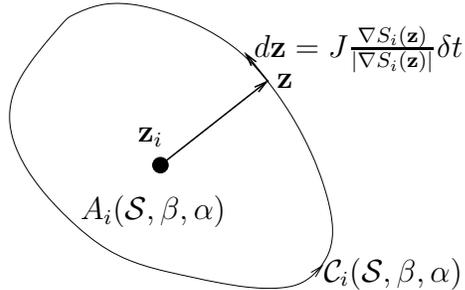}
\caption{Computation of the reception area of transmitter $i$.
\label{local_capacity_fig:snr_gradient}}
\end{figure}

Figure \ref{local_capacity_fig:snr_gradient} is the figurative representation of the reception area $A_i({\cal S},\beta,\alpha)$ of an arbitrary transmitter $i$. 

We can formulate our problem as follows.
\begin{align}
\text{Maximize: } D_i({\bf z})&=\|{\bf z}-{\bf z}_i\|=\|{\bf z}\|, \notag \\
\text{subject to: } S_i({\bf z})&=\beta~,
\end{align}
where $D_i({\bf z})$ represents the Euclidean distance of point ${\bf z}$ from ${\bf z}_i$. The contour line of the function 
$$
D_i({\bf z})=\|{\bf z}\|=\kappa~,
$$ 
where $\kappa$ is a constant, is a circle in the plane with point ${\bf z}$ on its boundary at a given distance $\kappa$ from the origin. Similarly, the contour line of the SIR function 
\begin{equation}
S_i({\bf z})=\frac{\|{\bf z}-{\bf z}_i\|^{-\alpha}}{\sum\limits_{j\neq i,{\bf z}_j\in {\cal S}}\|{\bf z}-{\bf z}_j\|^{-\alpha}}=\frac{\|{\bf z}\|^{-\alpha}}{\sum\limits_{j\neq i,{\bf z}_j\in {\cal S}}\|{\bf z}-{\bf z}_j\|^{-\alpha}}=\beta~,
\label{range_eq:sinr}
\end{equation}
is the closed curve ${\cal C}_i({\cal S},\beta,\alpha)$ that forms the boundary of the reception area ${\cal A}_i({\cal S},,\beta,\alpha)$ of transmitter $i$. Our goal is to locate a point on the curve ${\cal C}_i({\cal S},\beta,\alpha)$ which is furthest from the location of the transmitter $i$ and to determine its Euclidean distance from ${\bf z}_i$. 

The contour lines of the functions $D_i({\bf z})$ and $S_i({\bf z})$ may be distinct but they will intersect or meet each other. In other words, while moving along the contour line of the function \mbox{$S_i({\bf z})=\beta$}, {\it i.e.}, along the curve ${\cal C}_i({\cal S},\beta,\alpha)$, the value of $D_i({\bf z})$ may vary. Only when the contour lines of the functions $D_i({\bf z})$ and $S_i({\bf z})$ meet tangentially, {\it i.e.}, when the contour lines meet but do not cross each other that we do not increase or decrease the value of $D_i({\bf z})$. The contour lines meet at the critical point where the tangent vectors of these contour lines are parallel and that point in the plane is the point of our interest. Note that, as the gradient of a function is perpendicular to the contour lines, the gradients of the functions $D_i({\bf z})$ and $S_i({\bf z})$ are also parallel at our point of interest. 

We define the gradient of the function $D_i({\bf z})$, $\nabla D_i({\bf z})$, as
$$
\nabla D_i({\bf z})=\left[
\begin {array}{c}
\frac{\partial}{\partial x}D_i({\bf z})\\
\noalign{\medskip}
\frac{\partial}{\partial y}D_i({\bf z})
\end {array}
\right]~.
$$

Similarly, the gradient of $S_i({\bf z})$, $\nabla S_i({\bf z})$, at point ${\bf z}$ on the curve ${\cal C}_i({\cal S},\beta,\alpha)$, is
$$
\nabla S_i({\bf z})=\left[
\begin {array}{c}
\frac{\partial}{\partial x}S_i({\bf z})\\
\noalign{\medskip}
\frac{\partial}{\partial y}S_i({\bf z})
\end {array}
\right]~.
$$ 
Deriving closed form expressions for the SIR function $S_i({\bf z})$ and its gradient function when simultaneous transmitters are positioned in a specific grid pattern in the plane does not seem tractable. Therefore, we  will obtain a discretized and numerically convergent representation of ${\cal C}_i({\cal S},\beta,\alpha)$ by finite elements.
The vector $d{\bf z}$, in Fig. \ref{local_capacity_fig:snr_gradient}, is co-linear with $J\frac{\nabla S_{i}({\bf z})}{|\nabla S_{i}({\bf z})|}$ where $J$ is the anti-clockwise rotation of $3\pi/2$ (or clockwise rotation of $\pi/2$) given by
$$
J=\left[\begin{array}{cc}
0 & 1\\
-1 & 0\end{array}\right]~.
$$ 
Therefore, we can fix 
$$
d{\bf z}=J\frac{\nabla S_{i}({\bf z})}{|\nabla S_{i}({\bf z})|}\delta t~,
$$
where $\delta t$ is assumed to be infinitesimally small. The sequence of points ${\bf z}(k)$ computed as 
\begin{align}
{\bf z}(0) & ={\bf z}\notag\\
{\bf z}(k+1) & ={\bf z}(k)+J\frac{\nabla S_{i}({\bf z}(k))}{|\nabla S_{i}({\bf z}(k))|}\delta t~,
\label{local_capacity_eq:parameterized}
\end{align} 
gives a discretized and numerically convergent parametric representation of ${\cal C}_i({\cal S},\beta,\alpha)$ by finite elements. 

We will use this discretized representation of ${\cal C}_i({\cal S},\beta,\alpha)$ and compute the gradients of the functions $D_i({\bf z}(k))$ and $S_i({\bf z}(k))$ at the sequence of points ${\bf z}(k)$ to determine the point where these gradients are parallel. This point will give the maximum transmission range of the transmitter $i$.  

\section{Evaluation and Results}
\label{range_sec:results}

In this section, we will perform detailed analysis of the MAC schemes, we have discussed in this article, under {\em no fading} channel model only. 

\subsection{Slotted ALOHA Scheme}

\begin{figure*}[!t]
\centering
\psfrag{range-m-probability}{}
\subfloat[$\beta$ is varying and $\alpha=4.0$.]
{
	\hspace{-1.25cm}
\includegraphics[scale=.9]{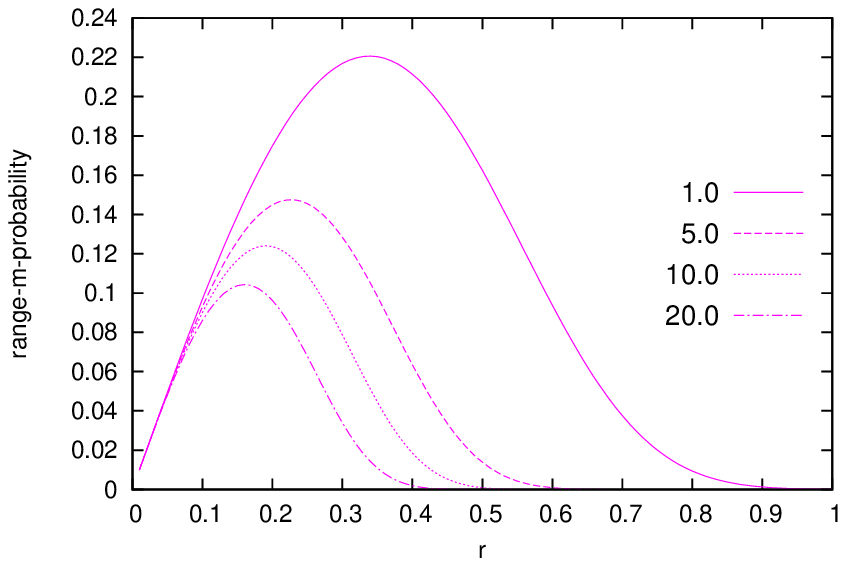}
}
\subfloat[$\beta=10.0$ and $\alpha$ is varying.]
{
	\hspace{-1.25cm}
\includegraphics[scale=.9]{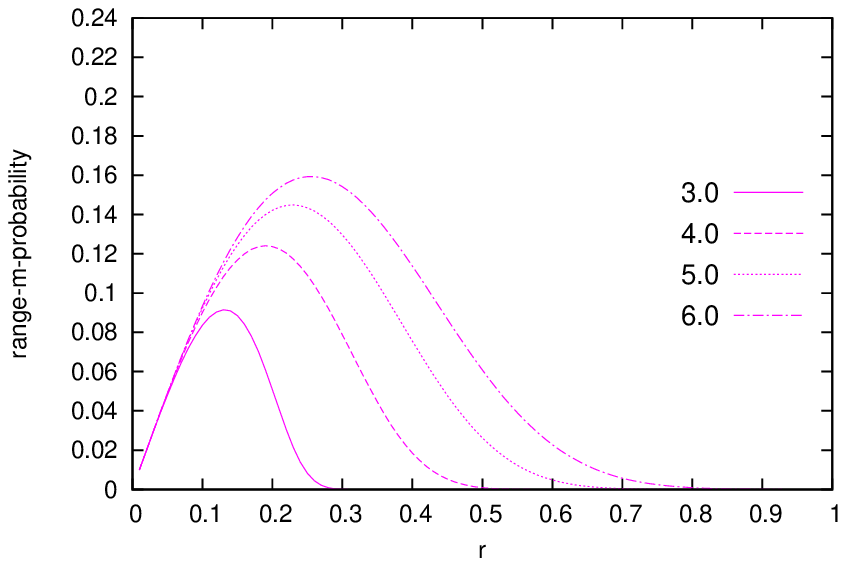}
}
\caption{Optimizing $rp(\lambda, r, \beta, \alpha)$ of slotted ALOHA scheme with $\lambda=1$.
\label{range_fig:op_aloha_}}
\end{figure*}

In case of slotted ALOHA scheme, our goal is to identify the optimal value of $r$ which maximizes the value of $rp(\lambda,r,\beta,\alpha)$, which is computed from \eqref{range_eq:prob_aloha} with varying values of the parameters $\beta$ and $\alpha$ and \mbox{$\lambda=1$}. 

Figure \ref{range_fig:op_aloha_} shows the plots of $rp(\lambda, r, \beta, \alpha)$ versus $r$, $\beta$ and $\alpha$ with \mbox{$\lambda=1$}. In Fig. \ref{range_fig:op_aloha_}(a), $r$ and $\beta$ are varying and $\alpha$ is fixed at $4.0$. Similarly, in Fig. \ref{range_fig:op_aloha_}(b), $r$ and $\alpha$ are varying and $\beta$ is fixed at $10.0$. Note that, as we are only interested in the optimal value of $r$ which maximizes the value of $rp(\lambda,r,\beta,\alpha)$ at given $\beta$ and $\alpha$ and $\lambda=1$, we only refer to this value by $r_1$ in our discussion.

\subsection{Grid Pattern Based Medium Access Control Schemes}
\label{range_sec:grid_pattern}

In case of grid pattern based schemes, we will use the numerical method described in \S \ref{range_sec:tx_range_method} to compute the maximum transmission range.  In order to approach an infinite map, we will perform numerical simulations in a very large network spread over $2D$ square map with length of each side equal to $10000$ meters. The transmitters are spread in this network area in square, rectangular, hexagonal or triangular grid pattern. We set the parameter $d$, for all grid patterns, equal to $25$ meters although it will have no effect on the validity of our conclusions as normalized maximum transmission range $r_1$ is independent of $\lambda$. To keep away border effects, we will only compute the maximum transmission range of the transmitter $i$, located in the center of the network area, {\it i.e.}, \mbox{${\bf z}_i=(x_i,y_i)=(0,0)$}. The network area is large enough so that the reception area of transmitter $i$, and hence its maximum transmission range, is close to its reception area and maximum transmission range in an infinite map. In case of rectangular grid pattern, we vary the values of the factors $k_1$ and $k_2$ in the ratio of $\frac{1}{2}$ and $\frac{1}{4}$. Note that the factors $k_1$ and $k_2$ are associated with the construction of the rectangular grid pattern as shown in Fig. \ref{schemes_fig:grid_layouts}. This allows us to derive the conclusions on the impact of these factors on the maximum transmission range of rectangular grid pattern based schemes. We can determine the value of $\lambda$ using the Voronoi tessellation of the network map and we derive its value for square, rectangular, hexagonal and triangular grid pattern to be equal to $\frac{1}{d^2}$, $\frac{1}{k_1k_2d^2}$, $\frac{4}{3\sqrt{3}d^2}$ and $\frac{2}{\sqrt{3}d^2}$ respectively. 

\subsection{Summary of Results}
\label{range_sec:summary_results}

\begin{figure*}[!t]
\centering
\subfloat[$\beta$ is varying and $\alpha=4.0$.]
{
	\hspace{-1.25cm}
\includegraphics[scale=.9]{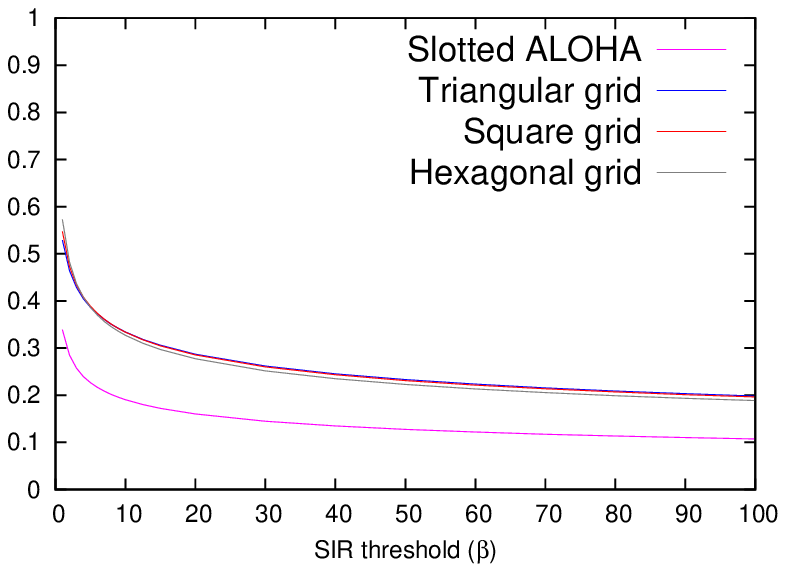}
}
\subfloat[$\beta$ is varying and $\alpha=4.0$. Rectangular grid with given $(k_1:k_2)$ ratio.]
{
	\hspace{-1.25cm}
\includegraphics[scale=.9]{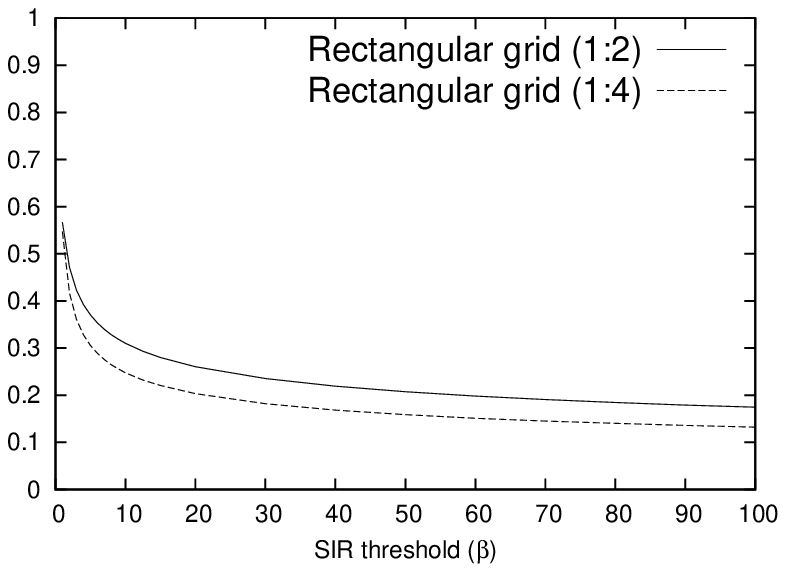}
}

\subfloat[$\beta=10.0$ and $\alpha$ is varying.]
{
	\hspace{-1.25cm}
\includegraphics[scale=.9]{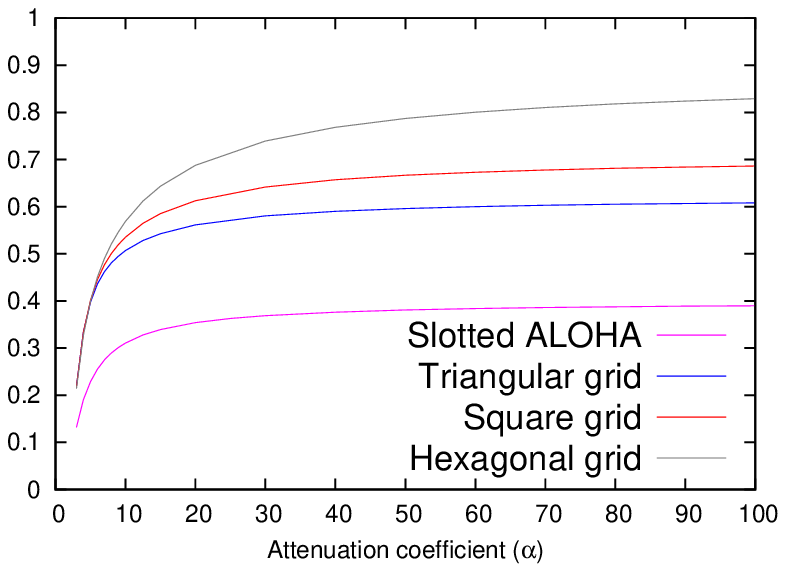}
}
\subfloat[$\beta=10.0$ and $\alpha$ is varying. Rectangular grid with given $(k_1:k_2)$ ratio.]
{
	\hspace{-1.25cm}
\includegraphics[scale=.9]{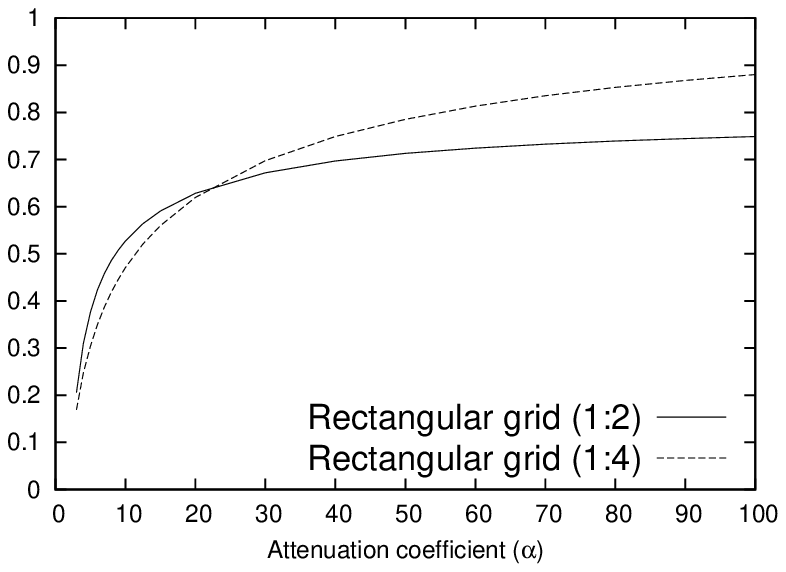}
}
\caption{Optimal $r_1=\sqrt{\lambda}r_{\lambda}$ of slotted ALOHA and grid pattern based schemes with $\lambda=1$.
\label{range_fig:range_grids_}}
\end{figure*}

\begin{figure*}[!t]
\centering
\subfloat[$\beta$ is varying and $\alpha=4.0$.]
{
	\hspace{-1.25cm}
\includegraphics[scale=.9]{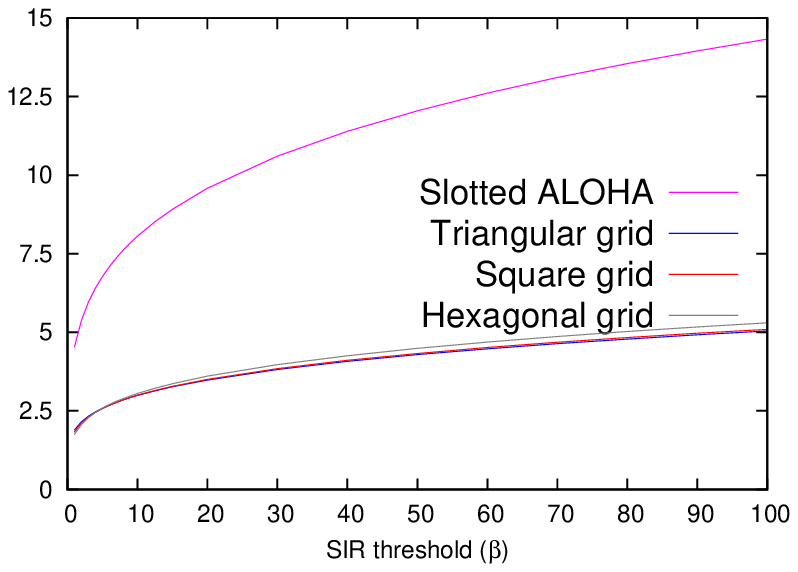}
}
\subfloat[$\beta$ is varying and $\alpha=4.0$. Rectangular grid with given $(k_1:k_2)$ ratio.]
{
	\hspace{-1.25cm}
\includegraphics[scale=.9]{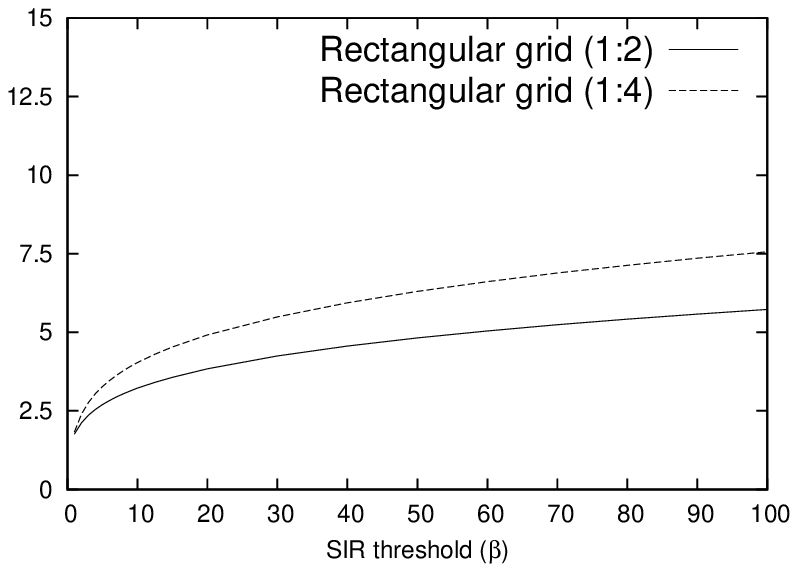}
}

\subfloat[$\beta=10.0$ and $\alpha$ is varying.]
{
	\hspace{-1.25cm}
\includegraphics[scale=.9]{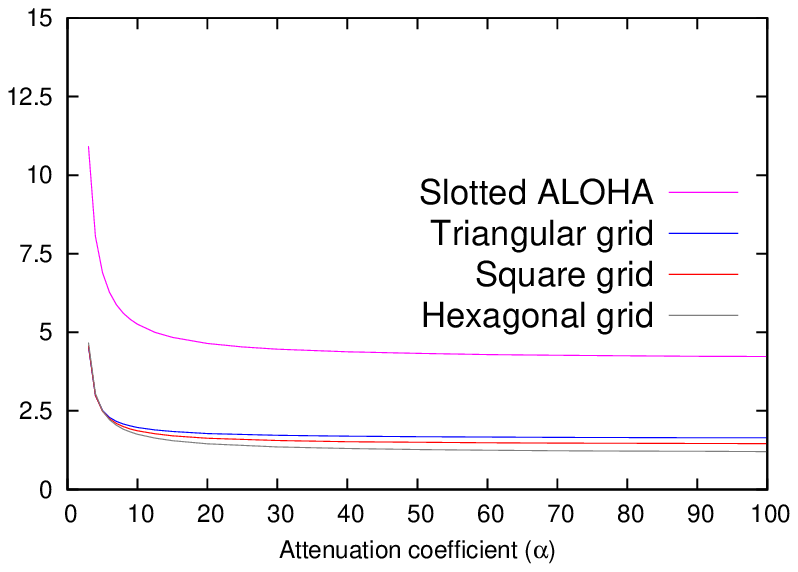}
}
\subfloat[$\beta=10.0$ and $\alpha$ is varying. Rectangular grid with given $(k_1:k_2)$ ratio.]
{
	\hspace{-1.25cm}
\includegraphics[scale=.9]{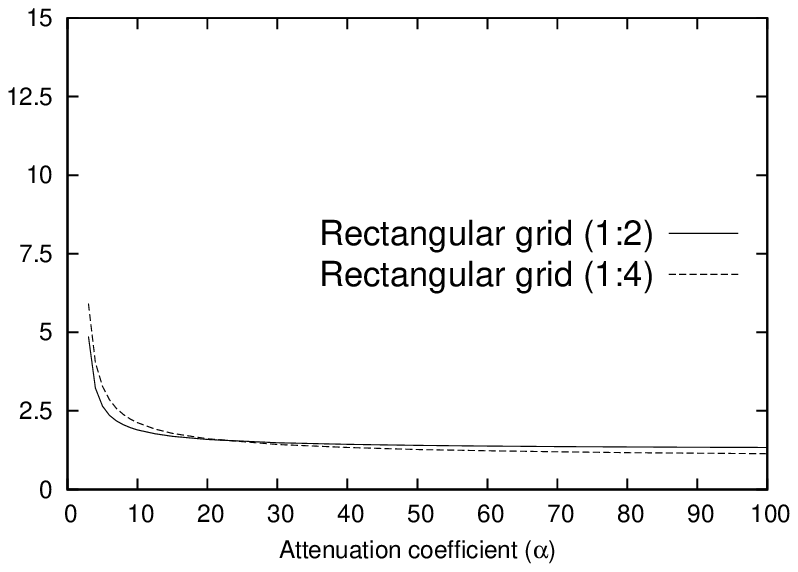}
}
\caption{Optimal (minimized) $\frac{1}{rp(\lambda, r, \beta, \alpha)}$ of slotted ALOHA and grid pattern based schemes with $\lambda=1$.
\label{range_fig:op_rp_grids_}}
\end{figure*}

\begin{figure*}[!t]
\centering
\subfloat[$\beta$ is varying and $\alpha=4.0$.]
{
	\hspace{-1.25cm}
\includegraphics[scale=.9]{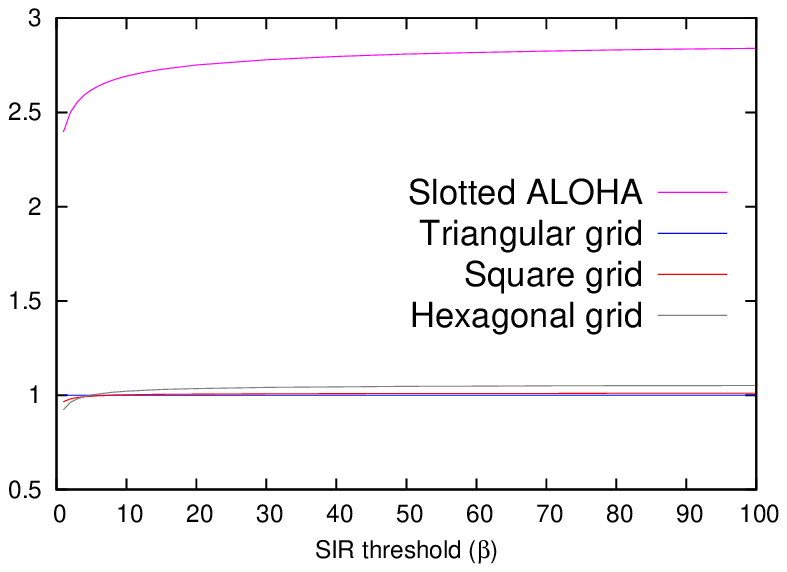}
}
\subfloat[$\beta$ is varying and $\alpha=4.0$. Rectangular grid with given $(k_1:k_2)$ ratio.]
{
	\hspace{-1.25cm}
\includegraphics[scale=.9]{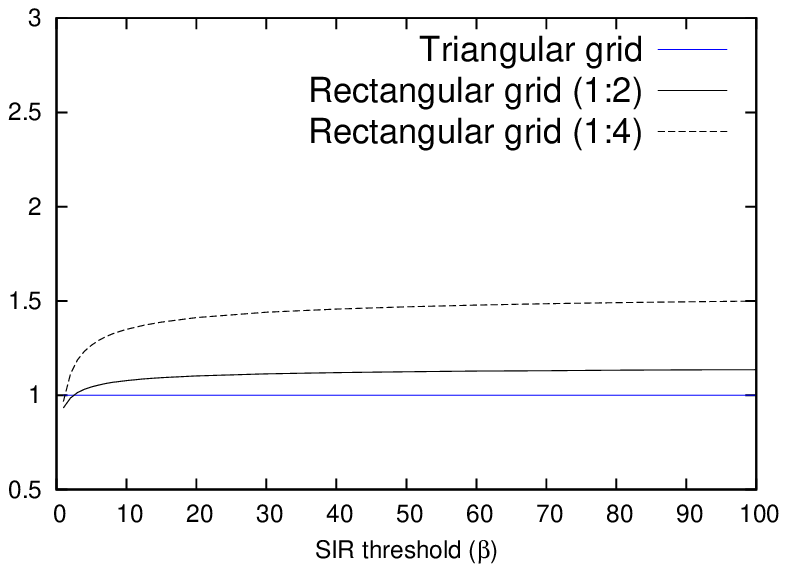}
}

\subfloat[$\beta=10.0$ and $\alpha$ is varying.]
{
	\hspace{-1.25cm}
\includegraphics[scale=.9]{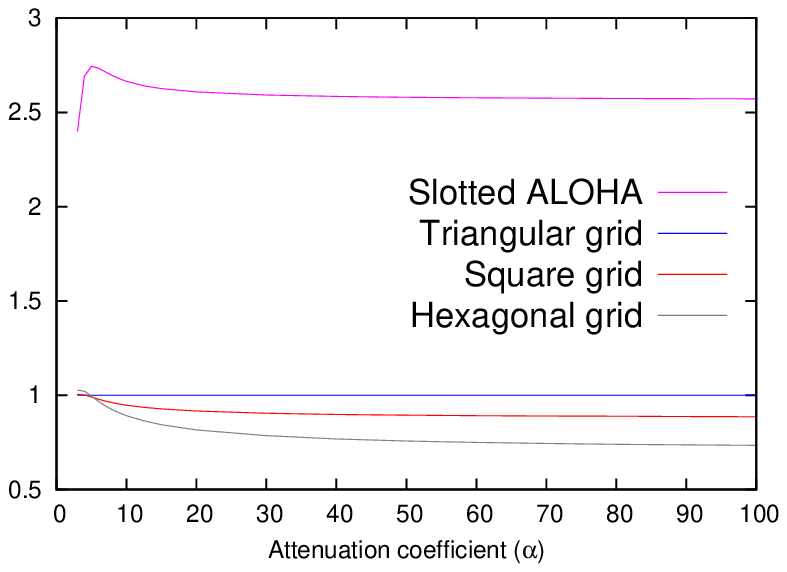}
}
\subfloat[$\beta=10.0$ and $\alpha$ is varying. Rectangular grid with given $(k_1:k_2)$ ratio.]
{
	\hspace{-1.25cm}
\includegraphics[scale=.9]{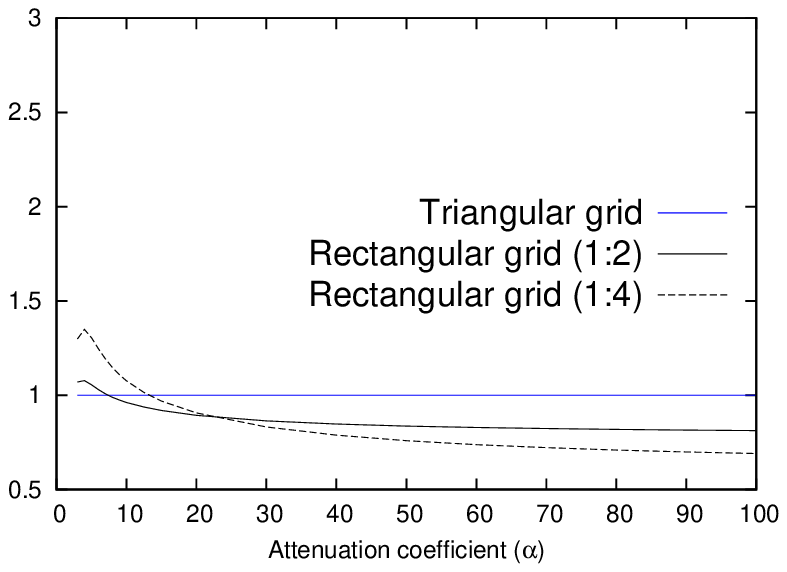}
}
\caption{Optimal (minimized) $\frac{1}{rp(\lambda, r, \beta, \alpha)}$ of slotted ALOHA and grid pattern based schemes normalized {\em w.r.t.} triangular grid pattern based scheme with $\lambda=1$. 
\label{range_fig:op_rp_compare_}}
\end{figure*}

Figure \ref{range_fig:range_grids_} shows the plots of the normalized optimal transmission range $r_1$ with varying $\beta$ and $\alpha$ for all MAC schemes. In Fig. \ref{range_fig:range_grids_}(a) and \ref{range_fig:range_grids_}(b), $\beta$ is varying and $\alpha$ is fixed at $4.0$. Similarly, in Fig. \ref{range_fig:range_grids_}(c) and \ref{range_fig:range_grids_}(d), $\beta$ is fixed at $10.0$ and $\alpha$ is varying. These results show that under typical values of $\beta$, in the range of $1$ to $20$, and $\alpha$, in the range of $3$ to $6$, the normalized maximum transmission range of all grid pattern based schemes are almost similar and {\em at most} double the normalized optimal average transmission range of slotted ALOHA scheme. However, as the values of $\beta$ and $\alpha$ approach their extremities, certain grid pattern based schemes seem to outperform the others. 

For example, as $\beta$ approaches zero, the maximum transmission range is influenced most by the closest source of interference and the most optimal scheme which maximizes the transmission range, in this case, is based on rectangular grid pattern with \mbox{$k_2\gg k_1$}. However, as $\beta$ increases and approaches the value of $100$, influence from more distant transmitters increase and the shape of the reception area approaches the shape of a small circular disk with transmitter at the center. Therefore, in this case, the most optimal scheme to maximize the transmission range is based on triangular grid pattern. 

On the other hand, as $\alpha$ increases, the reception area is influenced most by the interference from the nearest interferer as compared to any other interferer and the area of correct reception tends to be the Voronoi cell. Therefore, in this case, the most optimal scheme should also be based on rectangular grid pattern with \mbox{$k_2\gg k_1$}. We will discuss the asymptotic behavior of normalized maximum transmission ranges of grid pattern based schemes in detail in \S \ref{sec:asymp}. 

Figure \ref{range_fig:op_rp_grids_} shows the plots of the optimized (minimized) number of retransmissions $\frac{1}{rp(\lambda,r,\beta,\alpha)}$ required to transport a packet through multiple relays from its source node to its destination node located at unit distance with \mbox{$\lambda=1$} and varying $\beta$ and $\alpha$. In Fig. \ref{range_fig:op_rp_grids_}(a) and \ref{range_fig:op_rp_grids_}(b), $\beta$ is varying and $\alpha$ is fixed at $4.0$ whereas in Fig. \ref{range_fig:op_rp_grids_}(c) and \ref{range_fig:op_rp_grids_}(d), $\beta$ is fixed at $10.0$ and $\alpha$ is varying. 

Note that, in case of grid pattern based schemes, $p(\lambda,r,\beta,\alpha)$ is equal to one as, under {\em no fading} channel model, the reception area of a transmitter remains the same {\em modulo} a rotation. Therefore, in case of grid pattern based schemes, \mbox{$\frac{1}{rp(\lambda,r,\beta,\alpha)}=\frac{1}{r}$}. In order to quantify the improvement in the quantity $\frac{1}{rp(\lambda,r,\beta,\alpha)}$ by grid pattern based schemes over slotted ALOHA scheme, we perform a scaled comparison of slotted ALOHA and all grid pattern based schemes which is obtained by dividing the optimized quantity $\frac{1}{rp(\lambda,r,\beta,\alpha)}$ of all these schemes with the minimized quantity $\frac{1}{rp(\lambda,r,\beta,\alpha)}$ of triangular grid pattern based scheme. The reason of choosing triangular grid pattern based scheme as the reference here is twofold. First is that, under typical values of $\beta$ and $\alpha$, it can achieve the most optimal normalized transmission range in the wireless network although it is almost similar to other grid pattern based schemes. Second reason is that, triangular grid pattern based scheme can be visualized as the most optimal node coloring scheme which ensures that neighboring transmitters are located at exactly the distance $d$ from each other. Figure \ref{range_fig:op_rp_compare_} shows the scaled comparison with $\beta$ and $\alpha$ varying. It can be observed that the grid pattern based schemes can reduce the number of transmissions required to deliver a packet to its destination to {\em at most} one-third of the optimal average number of transmissions required with slotted ALOHA scheme. In other words, capacity in wireless networks with grid pattern based schemes can {\em at most} be three times the network capacity with slotted ALOHA scheme. 

\section{Asymptotic Analysis of Grid Pattern Based Medium Access Control Schemes}
\label{sec:asymp}

In this section, we will analyze the asymptotic behavior of the normalized maximum transmission range of grid pattern based schemes as $\beta$ and $\alpha$ approach extreme values.

The unexpected result is that when $\beta$ becomes very small and approaches zero or $\alpha$ becomes very large and approaches infinity, the most optimal grid pattern to maximize the normalized transmission range is the {\em linear} pattern, {\it i.e.}, when the transmitters are equidistantly spaced on a line.

\subsection{With Respect to SIR Threshold}

In this section, we will again assume that the transmitter $i$ is located at origin and our aim is to find its maximum transmission range as $\beta$ approaches infinity.

From \eqref{range_eq:sinr}, we can write
$$
\frac{\|{\bf z}-{\bf z}_i\|^{-\alpha}}{\sum\limits_{j\neq i,{\bf z}_j\in {\cal S}}\|{\bf z}-{\bf z}_j\|^{-\alpha}}=\frac{\|{\bf z}\|^{-\alpha}}{\sum\limits_{j\neq i,{\bf z}_j\in {\cal S}}\|{\bf z}-{\bf z}_j\|^{-\alpha}}=\frac{r_{\bf z}^{-\alpha}}{I(x,y)}=\beta~,
$$
where $r_{\bf z}$ is the distance of point ${\bf z}$ from the location of the transmitter $i$ and $I(x,y)$ is the signal level of the interference at point ${\bf z}$.
As $\beta$ approaches infinity, the reception area of a transmitter shrinks to a circular disk of an infinitesimally small radius with the transmitter at its center, {\it i.e.}, the transmission radius of transmitter $i$ also shrinks to an infinitesimally small value. 

Therefore, we can write the above equation as
$$
\frac{r_{0}(0,0)^{-\alpha}}{\beta}=I(0,0)~,
$$
where $r_{0}(0,0)$ is the transmission range of transmitter $i$ at origin when $\beta$ approaches infinity and $I(0,0)$ is the received signal level, at the location of transmitter $i$, from all transmitters except $i$, {\it i.e.},
$$
I(0,0)=\sum_{j\neq i,{\bf z}_j\in {\cal S}}\|{\bf z}_i-{\bf z}_j\|^{-\alpha}=\sum_{j\neq i,{\bf z}_j\in {\cal S}}\|{\bf z}_j\|^{-\alpha}~.
$$

Our goal is to determine the {\em normalized} maximum transmission range $r'_1$ which is defined as
$$
r'_1=\beta^{\frac{1}{\alpha}}r_0(0,0)=\left(I(0,0)\right)^{-\frac{1}{\alpha}}~.
$$
Note that, in the above expression, we also normalize the transmission range with respect to ({\em w.r.t.}) $\beta$ as, when $\beta$ approaches infinity, the transmission range $r_0(0,0)$ will also approach zero. Therefore, we will compute $\beta^{\frac{1}{\alpha}}r_0(0,0)$ for various grid pattern based schemes in the same network setting, as described in \S \ref{range_sec:grid_pattern}, and identity the scheme which maximizes the normalized maximum transmission range in the wireless network under the asymptotic condition when $\beta$ approaches infinity. Table \ref{range_tbl:eigen_values} shows the values of the normalized maximum transmission range of various grid pattern based schemes as $\beta$ approaches infinity with \mbox{$\alpha=4.0$} and \mbox{$\lambda=1.0$}. It is evident that as $\beta$ approaches infinity, the most optimal maximum transmission range is obtained with triangular grid pattern based scheme. 

\begin{table}
\begin{center}
\begin{tabular}{|c|c|c|}
\cline{3-3}
\multicolumn{2}{c|}{} & \\
\multicolumn{2}{c|}{} & {\em $r'_1=\beta^{\frac{1}{\alpha}}r_0(0,0)$} \\
\multicolumn{2}{c|}{} & \\
   \hline

\multicolumn{2}{|c|}{} & \\
\multicolumn{2}{|c|}{{\em Square}}        					& $0.638232$ \\
\multicolumn{2}{|c|}{} & \\
    \hline
  & & \\
\multirow{5}{*}{\em Rectangular}  & $\frac{k_1}{k_2}=\frac{1}{2}$  	& $0.554905$ \\
  & & \\
    \cline{2-3}
  & & \\
                                                    & $\frac{k_1}{k_2}=\frac{1}{4}$  	& $0.409452$ \\
  & & \\
       \hline
\multicolumn{2}{|c|}{} & \\
\multicolumn{2}{|c|}{{\em Hexagonal}}        					& $0.609856$ \\
\multicolumn{2}{|c|}{} & \\
       \hline
\multicolumn{2}{|c|}{} & \\
\multicolumn{2}{|c|}{{\em Triangular}}        					& $0.644845$ \\
\multicolumn{2}{|c|}{} & \\
       \hline
\end{tabular}
\end{center}
\caption{Normalized maximum transmission range $r'_1$ of various grid pattern based schemes as $\beta$ approaches infinity. $\alpha=4.0$ and $\lambda=1.0$.
\label{range_tbl:eigen_values}}
\end{table}

\begin{figure*}[!t]
\centering
\subfloat[$\beta=10^{-2}$.]
{
	\includegraphics[scale=1.1]{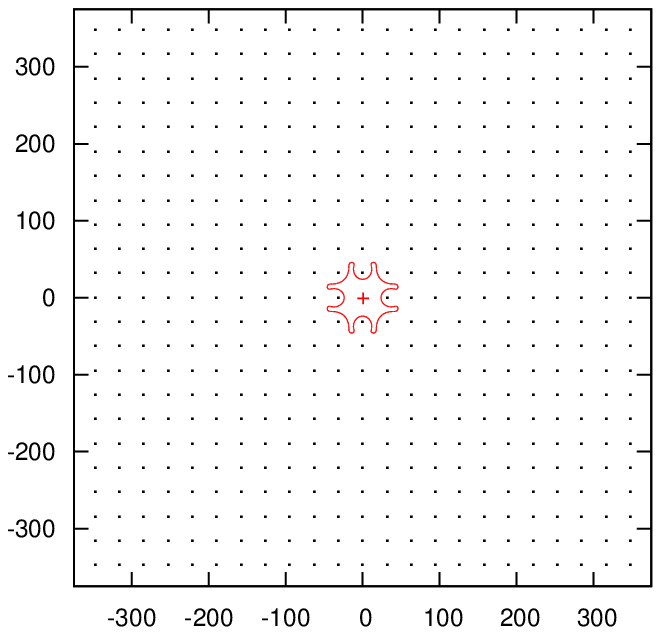}
}
\subfloat[$\beta=10^{-3}$.]
{
	\includegraphics[scale=1.1]{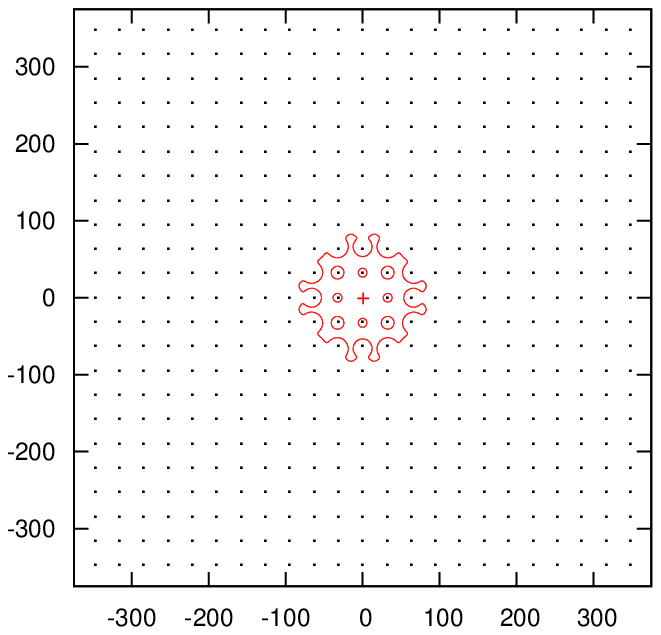}
}

\subfloat[$\beta=10^{-4}$.]
{
	\includegraphics[scale=1.1]{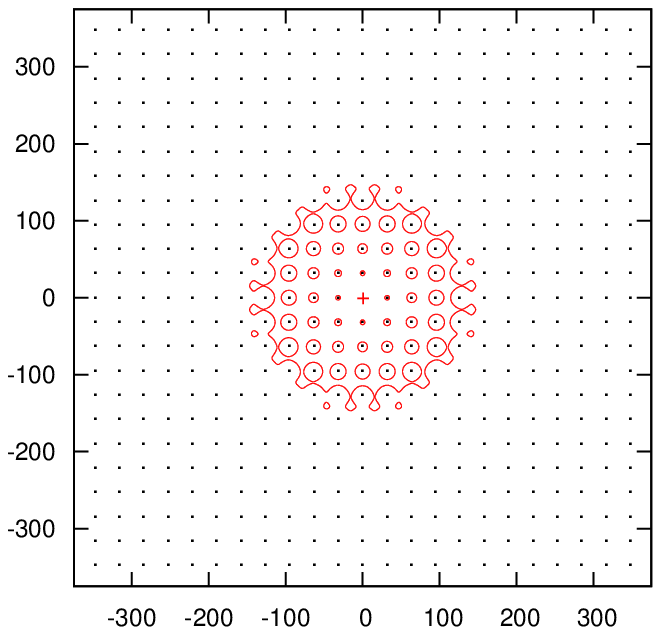}
}
\subfloat[$\beta=10^{-5}$.]
{
	\includegraphics[scale=1.1]{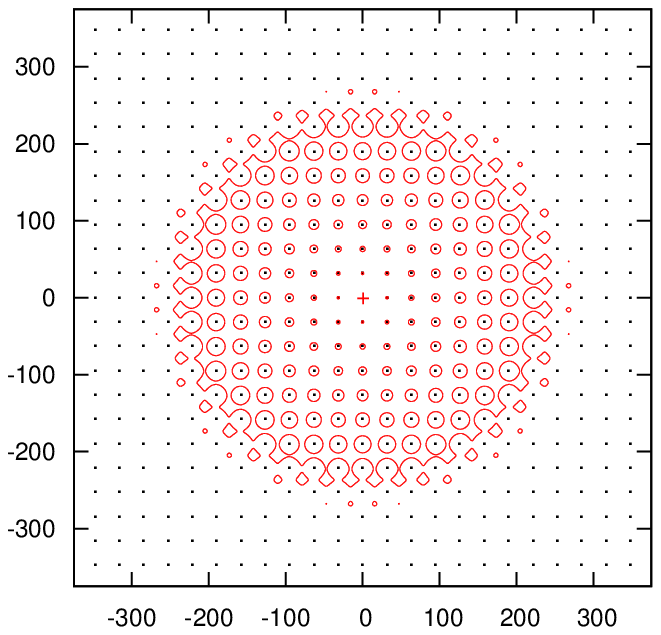}
}
\caption{Reception area of a transmitter with square grid pattern based scheme as $\beta$ approaches an infinitesimally small value. $\alpha=4.0$, $d=\sqrt{1000}$ and \mbox{$\lambda=\frac{1}{d^2}=\frac{1}{1000}$}.
\label{range_fig:sinr_less_1}}
\end{figure*}

\begin{figure*}[!t]
\centering
\subfloat[$\frac{k_1}{k_2}=\frac{10}{100}$.]
{
	\includegraphics[scale=1.1]{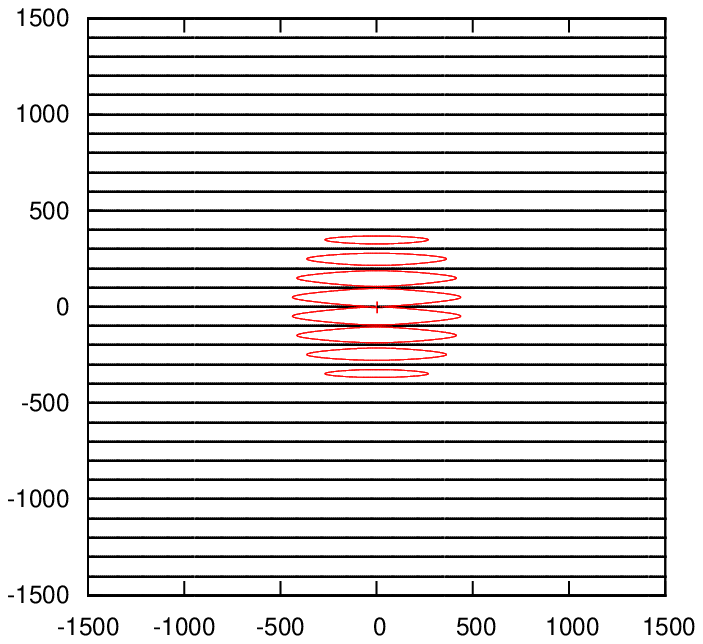}
}
\subfloat[$\frac{k_1}{k_2}=\frac{5}{200}$.]
{
	\includegraphics[scale=1.1]{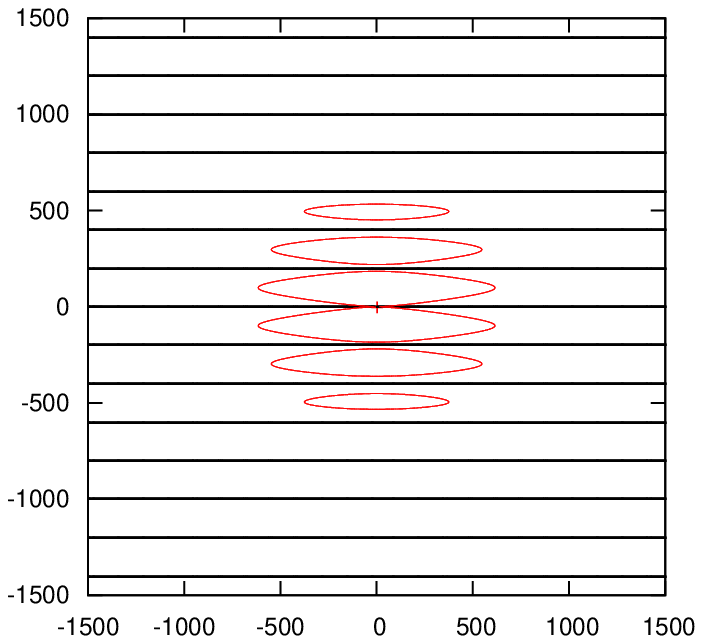}
}

\subfloat[$\frac{k_1}{k_2}=\frac{2.5}{400}$.]
{
	\includegraphics[scale=1.1]{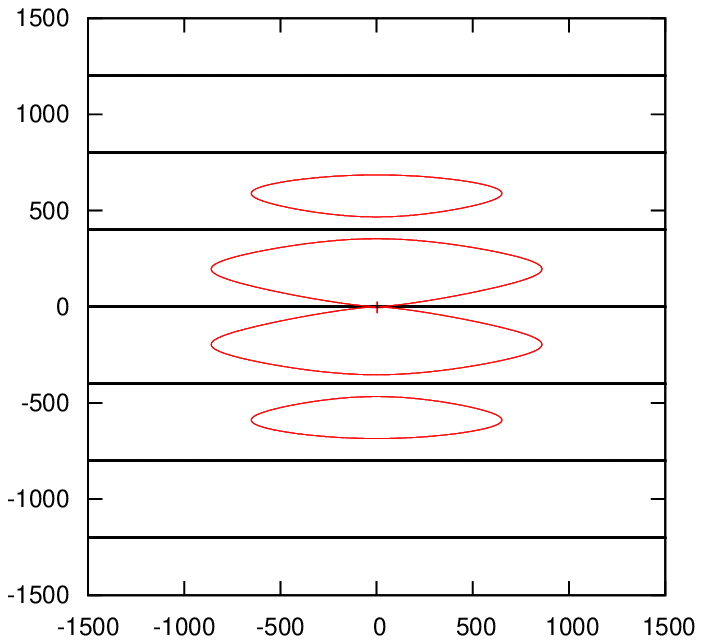}
}
\subfloat[$\frac{k_1}{k_2}=\frac{1}{1000}$.]
{
	\includegraphics[scale=1.1]{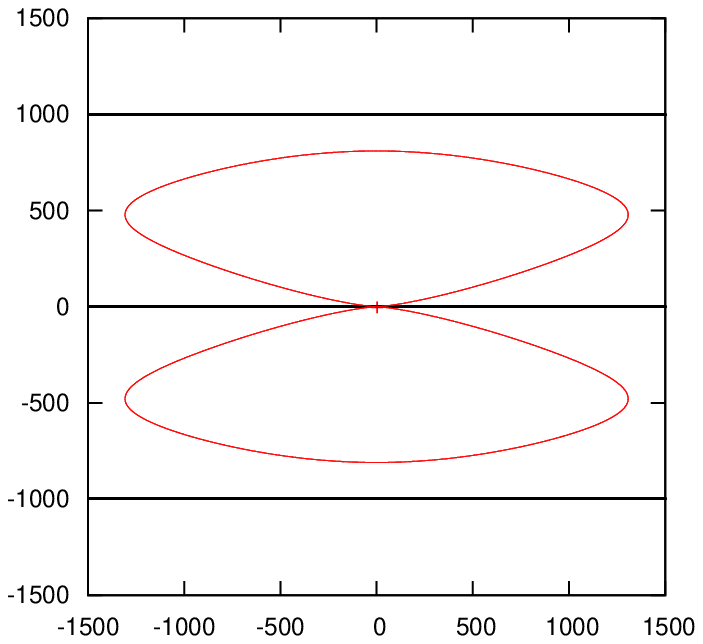}
}
\caption{Reception area of a transmitter with rectangular grid pattern based scheme as $\frac{k_1}{k_2}$ approaches an infinitesimally small value. $\beta=10^{-5}$, $\alpha=4.0$, $d=1$ and \mbox{$\lambda=\frac{1}{k_1k_2d^2}=\frac{1}{1000}$} is kept {\em constant}.
\label{range_fig:sinr_zero}}
\end{figure*}

It is also interesting to see what happens when $\beta$ approaches zero. 

We know that for $\beta$ equal or greater than one, the reception areas of the transmitters do not overlap as there is only one transmitter that can be received successfully at any location in the plane. Figure \ref{local_capacity_fig:reception_area} shows such an example where the reception areas of randomly distributed transmitters are shown for $\beta$ greater or equal to one. However, as the value of $\beta$ decreases and approaches an infinitesimally small value, the reception areas of the transmitters start to overlap. 

In Fig. \ref{range_fig:sinr_less_1}, nodes in the network use square grid pattern based scheme and this figure shows the transformation of the reception area of a transmitter (marked in red) in the same network setting, as described in \S \ref{range_sec:grid_pattern}, as the value of $\beta$ approaches an infinitesimally small value. Note that, as $\beta$ approaches an infinitesimally small value, reception area of the transmitter becomes {\em non-contiguous} and also excludes the locations of the interferers where the signal level of interference is infinite. We can see that as $\beta$ approaches zero, the reception area tends to an infinite plane though excluding the points of the locations of the interferers. 

Our results in Fig. \ref{range_fig:range_grids_}(b) and \ref{range_fig:op_rp_grids_}(b) also hint that as $\beta$ decreases, the most optimal transmission scheme to maximize the transmission range is based on rectangular grid pattern with $k_2>k_1$. Figure \ref{range_fig:sinr_zero} shows the transformation of the reception area of a transmitter when nodes, in the network setting of \S \ref{range_sec:grid_pattern}, use rectangular grid pattern based scheme and the value of ${k_1}/{k_2}$ approaches an infinitesimally small value. Note that because of a very small $k_1$, the transmitters in this figure appear to be densely packed on lines. It can be noticed that as the value of ${k_1}/{k_2}$ decreases, the reception area of the transmitter becomes {\em contiguous} and its maximum transmission range also increases. For example, at $\beta=10^{-5}$, the results in Fig. \ref{range_fig:sinr_less_1}(d) and \ref{range_fig:sinr_zero}(d) show that, as compared to square grid pattern, rectangular grid pattern based scheme can achieve longer transmission range. Note that, in Fig. \ref{range_fig:sinr_zero}, values of the factors $k_1$ and $k_2$ are varied such that $\lambda$ remains constant and we observed their impact on the maximum transmission range. From our results, we conjecture that as $\beta$ approaches zero, the most optimal MAC scheme to maximize the normalized transmission range in multi-hop wireless networks, such that the reception area of a transmitter also remains connected, is based on rectangular grid pattern with $k_2$ approaching infinity or, in other words, {\em linear} pattern with transmitters equidistantly placed on a straight line and, as the density of nodes in the network approach infinity, $k_1$ ({\it i.e.} the distance in-between simultaneous transmitters on the line) approaching an infinitesimally small value. 

\subsection{With Respect to Attenuation Coefficient}

As $\alpha$ approaches infinity, the reception area of a transmitter tends to be the Voronoi cell around this transmitter. This is true for all values of $\beta$. This is because the reception area is influenced most by the interference from the nearest transmitter as compared to any other transmitter. The average area of the Voronoi cell is equal to the inverse of the density of the set of simultaneous transmitters, {\it i.e.}, $1/\lambda$. Figure \ref{range_fig:asym_alpha} shows the reception area of a transmitter in the network setting of \S \ref{range_sec:grid_pattern} where nodes employ square, rectangular, hexagonal and triangular grid pattern based schemes with \mbox{$\beta=1.0$} and \mbox{$\alpha=100.0$}. In this figure, the arrows represent the directions of the maximum transmission range of a transmitter with each scheme. We have used the numerical method described in \S \ref{range_sec:tx_range_method} to plot these reception areas and we can see that at very high value of $\alpha$, the reception areas tend to be the Voronoi cell of a specific shape which depends on the positioning of simultaneous transmitters by the grid pattern based schemes. 

\begin{figure*}[!t]
\centering
\psfrag{r}{}
\subfloat[Square grid.]
{
	\includegraphics[scale=.8]{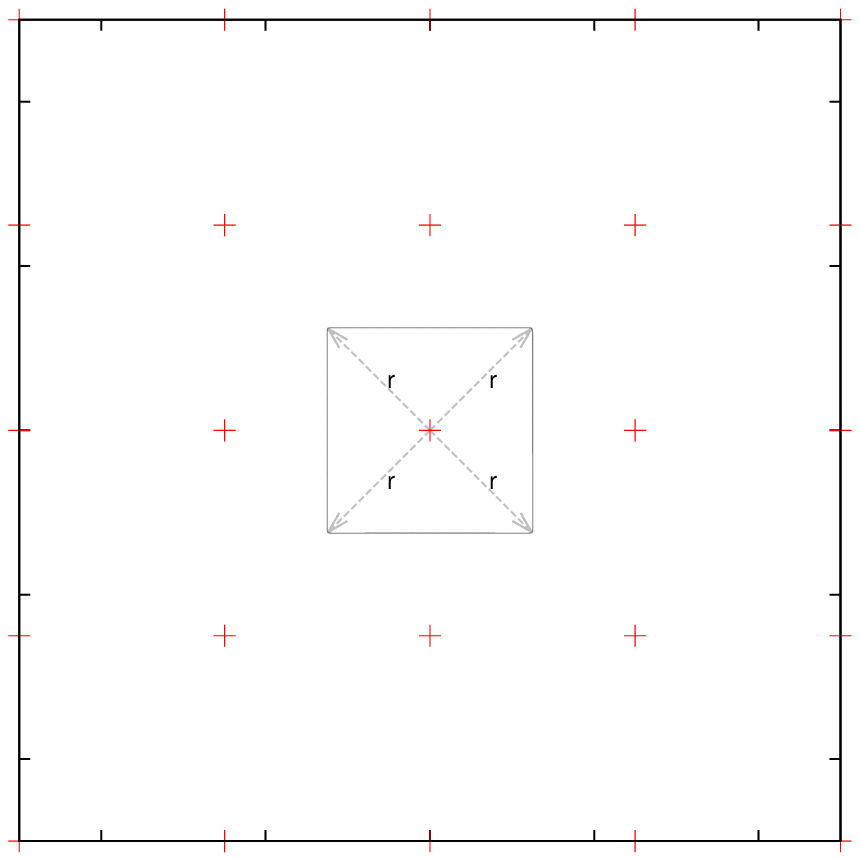}
}
\subfloat[Rectangular grid with $\frac{k_1}{k_2}=0.5$.]
{
	\includegraphics[scale=.8]{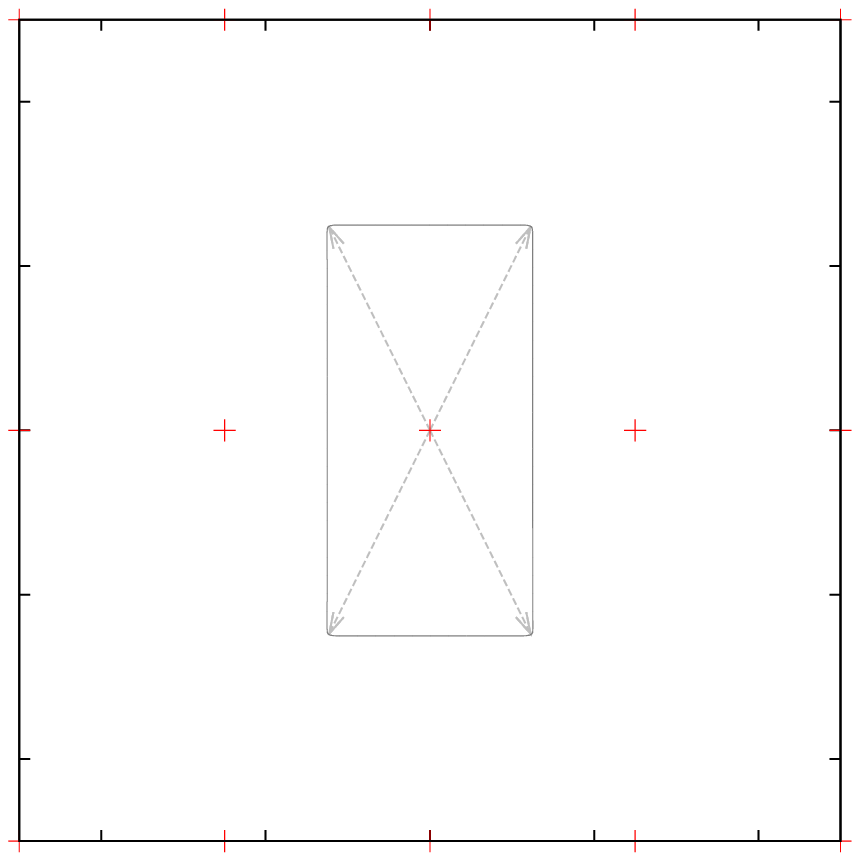}
}

\subfloat[Hexagonal grid.]
{
	\includegraphics[scale=.8]{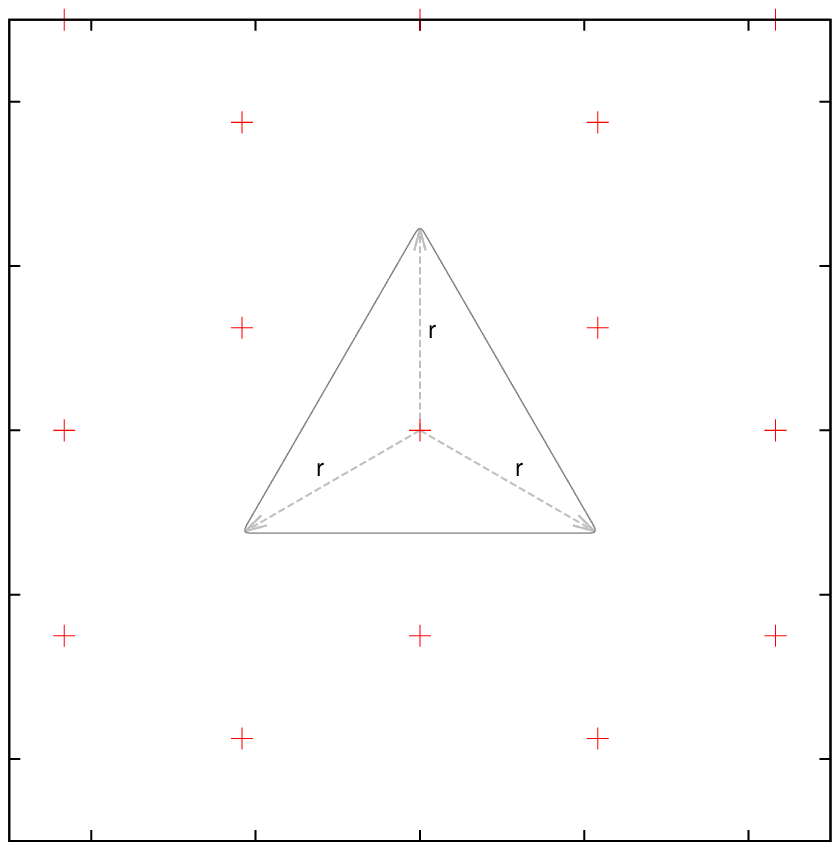}
}
\subfloat[Triangular grid.]
{
	\includegraphics[scale=.8]{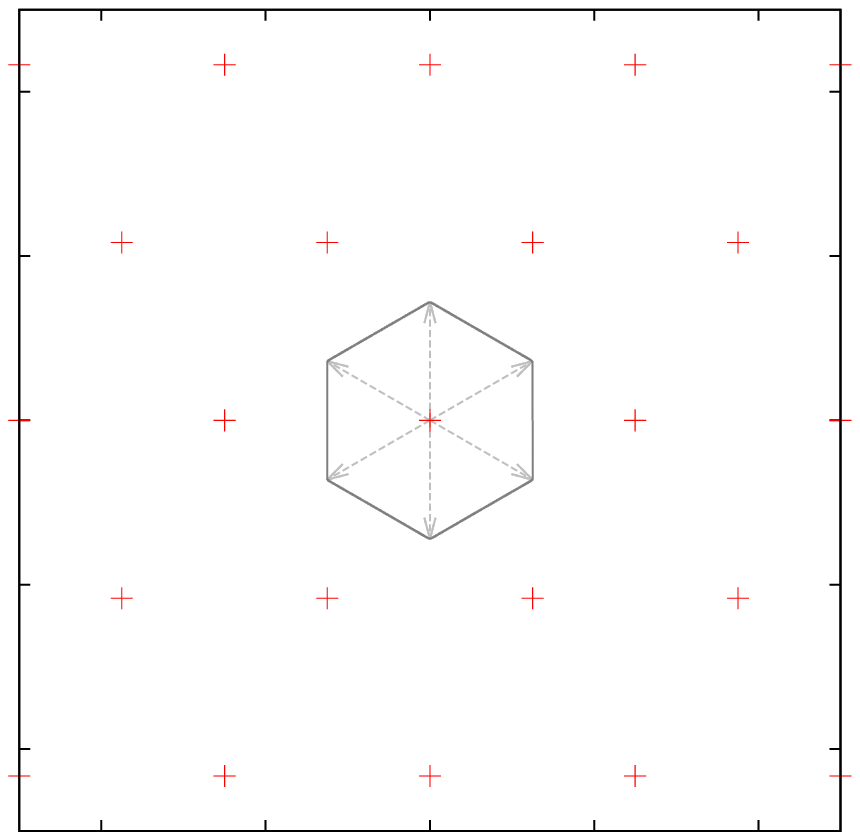}
}
\caption{Reception area of a transmitter with various grid pattern based schemes. Grids are constructed with parameter $d$. $\beta=1.0$ and $\alpha=100.0$.
\label{range_fig:asym_alpha}}
\end{figure*}

As $\alpha$ approaches infinity, reception area of a transmitter in square grid pattern based scheme approaches the shape of a square with area $d^2$ and this is irrespective of the value of $\beta$. This can also be observed in Fig. \ref{range_fig:asym_alpha}(a). From Voronoi tessellation of the network map, we can see that the density of the set of simultaneous transmitters $\lambda$ is equal to $1/d^2$ and the maximum transmission range is along the diagonals of the square reception area, given by \mbox{$r_{\lambda}=d/\sqrt{2}$}. Therefore, as $\alpha$ approaches infinity, the normalized maximum transmission range of a transmitter in square grid pattern based scheme approaches the value of
$$
r_1=\sqrt{\lambda}r_{\lambda}=\frac{1}{\sqrt{2}}\approx 0.707~.
$$
This result can also be verified by the plot in Fig. \ref{range_fig:range_grids_}(c) when $\alpha=100.0$. 

In case of rectangular grid pattern based scheme, as $\alpha$ approaches infinity, the reception area approaches the shape of a rectangle with area equal to $k_1k_2d^2$ and \mbox{$\lambda=1/k_1k_2d^2$} where $k_1$ and $k_2$ are the factors associated with the dimensions of the rectangular grid pattern as shown in Fig. \ref{schemes_fig:grid_layouts}. In this case also, the maximum transmission range is along the diagonals of the rectangular reception area and is given by 
\begin{align*}
r_{\lambda}&=\frac{d}{2}\sqrt{k_1^2+k_2^2}\\
&=\frac{\sqrt{k_1k_2}d}{2}\sqrt{\frac{\left(\frac{k_1}{k_2}\right)^2+1}{\left(\frac{k_1}{k_2}\right)}}~,
\intertext{and therefore,}
r_1&=\sqrt{\lambda}r_{\lambda}=\frac{1}{2}\sqrt{\frac{\left(\frac{k_1}{k_2}\right)^2+1}{\left(\frac{k_1}{k_2}\right)}}~.
\end{align*}
It is interesting to note that in case of rectangular grid pattern based scheme, the maximum transmission range depends on the factors $k_1$ and $k_2$. Note that, in case of square grid pattern, \mbox{$k_1/k_2=1$}. However, in case of rectangular grid pattern, $k_1/k_2$ can be any value less than one and as $k_2$ approaches infinity, $k_1/k_2$ approaches zero and $r_1$ approaches infinity. Our results in Fig. \ref{range_fig:range_grids_}(d) also verify that, with $\alpha$ approaching $100$, as $k_1/k_2$ decreases, $r_1$ increases. 

Similarly, we can also repeat the calculations and derive the maximum transmission ranges with hexagonal and triangular grid pattern based schemes under asymptotic conditions. 

Table \ref{range_tbl:asym_alpha} shows the normalized maximum transmission range of all grid pattern based schemes as $\alpha$ approaches infinity. The asymptotic behavior of maximum transmission range of various grid pattern based schemes shows that as $\alpha$ increases, the normalized maximum transmission range in the wireless network can be achieved when simultaneous transmitters are positioned in a rectangular grid pattern with \mbox{$k_2\gg k_1$}. This also allows us to conjecture that asymptotically, as $\alpha$ approaches infinity, the most optimal positioning of simultaneous transmitters, to maximize the normalized transmission range in multi-hop wireless networks, is in a rectangular grid pattern with $k_2$ approaching infinity or, in other words, {\em linear pattern} with transmitters equidistantly spaced on a straight line. Moreover, as the density of nodes in the network approaches infinity, $k_1$ (the distance in-between simultaneous transmitters on the line) approaches an infinitesimally small number. 

\begin{table}
\begin{center}
\begin{tabular}{|c|c|c|c|}
\cline{2-4}
\multicolumn{1}{c|}{} & & &  \\
\multicolumn{1}{c|}{} & {\em $\lambda$}     &  {\em $r_{\lambda}$}  & {\em $r_1=\sqrt{\lambda}r_{\lambda}$}\\
\multicolumn{1}{c|}{} & & &  \\
   \hline

\multicolumn{1}{|c|}{} & & &  \\
\multicolumn{1}{|c|}{{\em Square}}        & $\frac{1}{d^2}$ & $\frac{d}{\sqrt{2}}$  &    $\frac{1}{\sqrt{2}}\approx 0.707$   \\
\multicolumn{1}{|c|}{} & & &  \\
    \hline
\multicolumn{1}{|c|}{} & & &  \\
\multicolumn{1}{|c|}{\em Rectangular}  &  $\frac{1}{k_1k_2d^2}$  &  $\frac{\sqrt{k_1k_2}d}{2}\sqrt{\frac{\left(\frac{k_1}{k_2}\right)^2+1}{\left(\frac{k_1}{k_2}\right)}}$  &   $\frac{1}{2}\sqrt{\frac{\left(\frac{k_1}{k_2}\right)^2+1}{\left(\frac{k_1}{k_2}\right)}}$  \\
\multicolumn{1}{|c|}{$\left(\frac{k_1}{k_2}<1\right)$} & & & \\
\multicolumn{1}{|c|}{} & & &  \\
    \hline
\multicolumn{1}{|c|}{} & & &  \\
\multicolumn{1}{|c|}{{\em Hexagonal}}        &   $\frac{4}{3\sqrt{3}d^2}$  & $d$ &   $\frac{2}{\sqrt{3\sqrt{3}}}\approx 0.877$  \\
\multicolumn{1}{|c|}{} & & &  \\
       \hline
\multicolumn{1}{|c|}{} & & &  \\
\multicolumn{1}{|c|}{{\em Triangular}}         & $\frac{2}{\sqrt{3}d^2}$  & $\frac{d}{\sqrt{3}}$ &   $\sqrt{\frac{2}{3\sqrt{3}}}\approx 0.620$ \\
\multicolumn{1}{|c|}{} & & &  \\
       \hline
\end{tabular}
\end{center}
\caption{Maximum transmission range of a transmitter with various grid pattern based schemes as $\alpha$ approaches infinity.
\label{range_tbl:asym_alpha}}
\end{table}

\section{Impact of Fading on the Performance of Medium Access Control Schemes}
\label{sec:impact_fading}

So far, in our discussion, we have used the {\em no fading} channel model. In this section, we will discuss the impact of fading on the performance of slotted ALOHA and grid pattern based MAC schemes.

In case of fading, one cannot define ${\cal A}_{i}({\cal S},\beta,\alpha)$ deterministically. Therefore, we will define it as
$$
|{\cal A}_{i}({\cal S},\beta,\alpha)|=\int\Pr(\gamma_i({\bf z})\geq \beta\sum_{j\neq i,{\bf z}_j\in{\cal S}}\gamma_j({\bf z}))d{\bf z}^2~,
$$
where the integration is over the network area ${\cal A}$. 

Note that $|{\cal A}_i|$ is the size of the surface area of ${\cal A}_i$, \mbox{$\Pr(\gamma_i({\bf z})\geq \beta\sum_{j\neq i,{\bf z}_j\in{\cal S}}\gamma_j({\bf z}))$} is the probability that the signal from transmitter $i$ is received successfully at point ${\bf z}$ and we recall from \eqref{schemes_eq:channel_gain} that under fading
$$
\gamma_i({\bf z})=\frac{F_i({\bf z})}{\|\mathbf{z}-\mathbf{z}_{i}\|^{\alpha}}~,
$$
where $F_i({\bf z})$ is an {\em i.i.d.} random variable which represents the fading of signal from transmitter $i$ to receiver located at point ${\bf z}$. 

\subsection{Slotted ALOHA Scheme}
\label{sec:impact_fading_aloha}

\begin{theorem}
In case of slotted ALOHA scheme, under Rayleigh fading channel model, the probability of successfully transmitting a packet to a receiver at distance r is
\begin{align}
p(\lambda,r,\beta,\alpha)&=\Pr({\cal W}(\lambda)<xF)\notag\\
&=\underset{n\geq0}{\sum}\frac{(-C\lambda)^{n}}{n!}\frac{\sin(\pi n\gamma)}{\pi}\Gamma(n\gamma)\psi(-n\gamma)\left(\frac{r^{-\alpha}}{\beta}\right)^{-n\gamma}~,
\label{range_eq:prob_aloha_fading}
\end{align}
where $C=\pi\psi(\gamma)\Gamma(1-\gamma)$ and $\psi(s)=\E[F^s]$.
\end{theorem}

\begin{proof}
We recall that in case of slotted ALOHA, under {\em no fading} channel model, the Laplace transform of the density function of the random variable ${\cal W}({\bf z},{\cal S})$ which is the received interference because of Poisson distributed transmitters in the network is given by \eqref{local_capacity_eq:laplace}. We can follow the same steps to arrive at the expression of $\tilde{w}(\theta,\lambda)$ under {\em Rayleigh fading} channel model as
$$
\tilde{w}(\theta,\lambda)=\exp\left(-\lambda \pi \Gamma\left(1-\frac{2}{\alpha}\right)\psi\left(\frac{2}{\alpha}\right)\theta^{\frac{2}{\alpha}}\right)~,
$$
where $\psi(\frac{2}{\alpha})=\E[F^{\frac{2}{\alpha}}]$. Using similar techniques and the inverse Laplace transform, we can derive the following identity
$$
\Pr({\cal W}(\lambda)<xF)=\underset{n\geq0}{\sum}\frac{(-C\lambda)^{n}}{n!}\frac{\sin(\pi n\gamma)}{\pi}\Gamma(n\gamma)\psi(-n\gamma)x^{-n\gamma}~,
$$
where $C=\pi\psi(\gamma)\Gamma(1-\gamma)$, $\gamma=\frac{2}{\alpha}$ and $\psi(s)=\E[F^s]$.
\end{proof}
If we consider that the random variable $F$ is expressed as \mbox{$F=e^{-u}$}, where $u$ is a random variable, we get
$$
\psi(s)=\E[F^s]=\int e^{-xs} \phi_u(x) dx~,
$$
where $\phi_u(x)$ is the probability density of $u$. If we assume that $u$ is uniformly distributed in the interval $[-f,+f]$, we have
$$
\psi(s)=\E[F^s]=\int e^{-xs}\phi_u(x)dx=\frac{1}{2f}\int e^{-xs}dx=\frac{1}{fs}\sinh(f s)~.
$$

\begin{figure}[!t]
\centering
\psfrag{r}{$r$}
\includegraphics[scale=0.9]{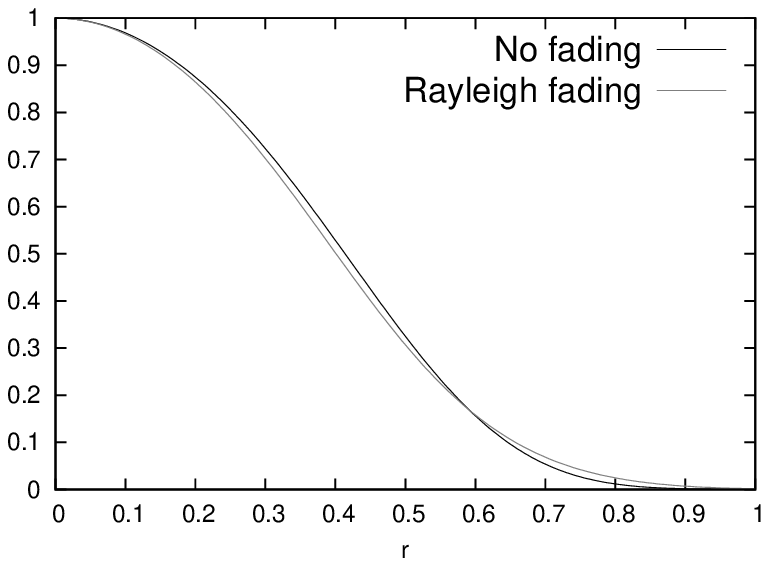}
\caption{\mbox{$p(1,r,\beta,\alpha)$} versus $r$ of slotted ALOHA under {\em no fading} and {\em Rayleigh fading} channel models. \mbox{$\lambda=1$}, \mbox{$\beta=1.0$} and \mbox{$\alpha=4.0$}.
\label{range_fig:dist_aloha_fading}}
\end{figure}

\begin{figure*}[!t]
\centering
\subfloat[$\beta$ is varying and $\alpha=4.0$.]
{
	\hspace{-1.25cm}
\includegraphics[scale=.9]{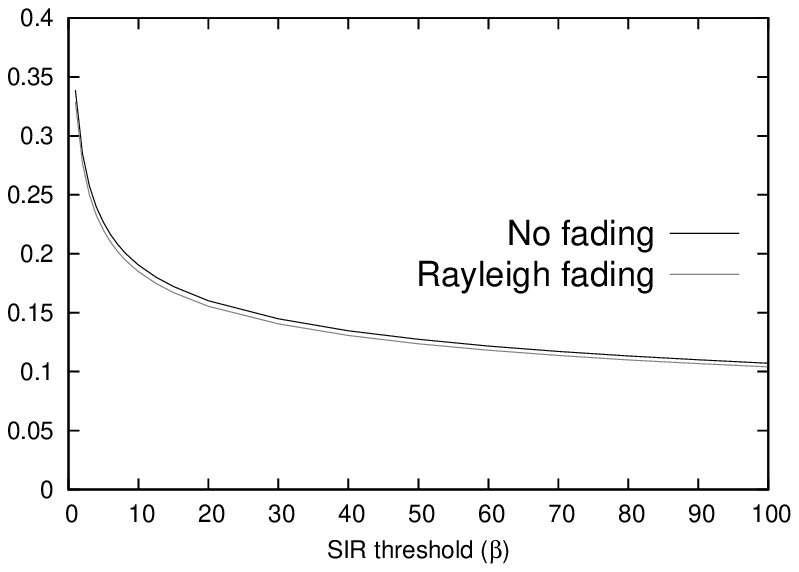}
}
\subfloat[$\beta=10.0$ and $\alpha$ is varying.]
{
	\hspace{-1.25cm}
\includegraphics[scale=.9]{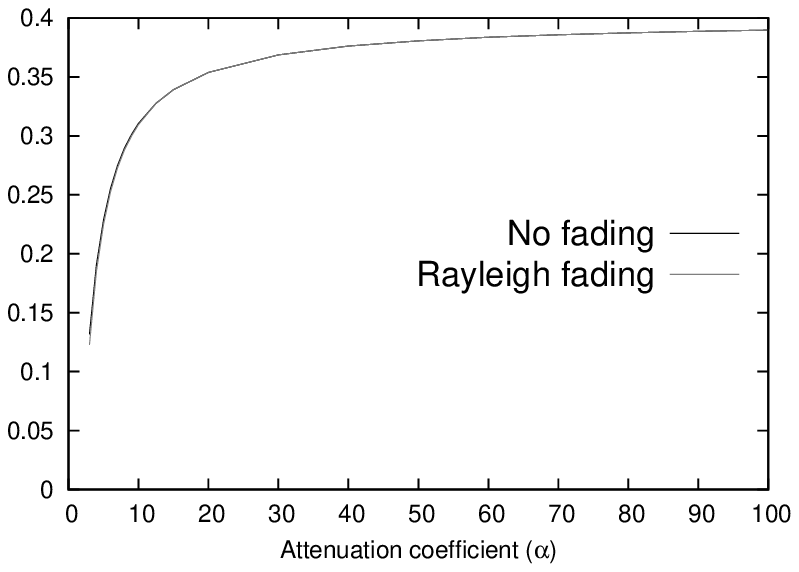}
}
\caption{Optimal $r_1=\sqrt{\lambda}r_{\lambda}$ of slotted ALOHA under {\em no fading} and {\em Rayleigh fading} channel models with $\lambda=1$.
\label{range_fig:range_aloha_fading}}
\end{figure*}

\begin{figure*}[!t]
\centering
\subfloat[$\beta$ is varying and $\alpha=4.0$.]
{
	\hspace{-1.25cm}
\includegraphics[scale=.9]{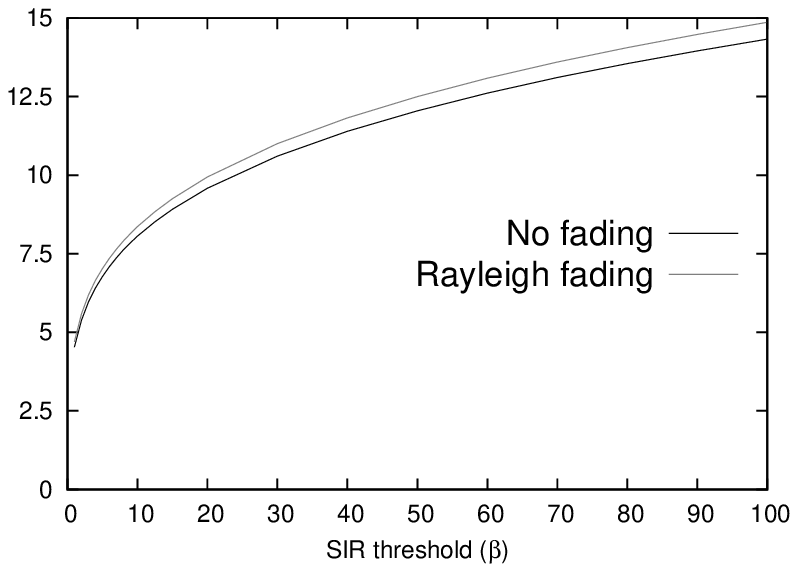}
}
\subfloat[$\beta=10.0$ and $\alpha$ is varying.]
{
	\hspace{-1.25cm}
\includegraphics[scale=.9]{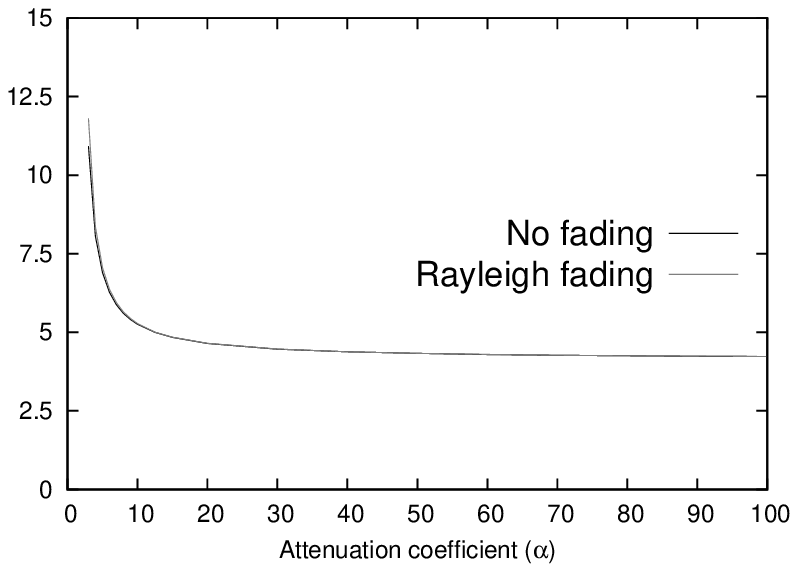}
}
\caption{Optimal (minimized) $\frac{1}{rp(\lambda, r, \beta, \alpha)}$ of slotted ALOHA under {\em no fading} and {\em Rayleigh fading} channel models with $\lambda=1$.
\label{range_fig:op_rp_aloha}}
\end{figure*}

Figure \ref{range_fig:dist_aloha_fading} shows the comparison of \mbox{$p(\lambda,r,\beta,\alpha)$} for slotted ALOHA under {\em no fading} and {\em Rayleigh fading} channel models  with $\lambda=1$, $\beta=1.0$ and $\alpha=4.0$. We recall that, in case of {\em no fading} channel model, \mbox{$F=1$} and \mbox{$p(\lambda,r,\beta,\alpha)$} is plotted using \eqref{local_capacity_eq:signal_pd} whereas, in case of {\em Rayleigh fading} channel model, \mbox{$F=e^u$} where $u$ is uniformly distributed in the interval $[-1,+1]$ and \mbox{$p(\lambda,r,\beta,\alpha)$} is plotted using \eqref{range_eq:prob_aloha_fading}. At smaller distances, the probability of successfully transmitting a packet is better under {\em no fading} as compared to {\em Rayleigh fading} channel model. However, as the value of $r$ increases, it is interesting to note that the situation is reversed and \mbox{$p(\lambda,r,\beta,\alpha)$} becomes slightly higher under {\em Rayleigh fading} channel model.

Figure \ref{range_fig:range_aloha_fading} shows the optimum transmission range of slotted ALOHA under {\em no fading} and {\em Rayleigh fading} channel models. It can be noticed in Fig. \ref{range_fig:range_aloha_fading}(a) that at typical value of $\alpha$ equal to $4.0$ and $\beta$ varying, the optimum transmission range is almost $3\%$ lower under {\em Rayleigh fading} channel model as compared to {\em no fading}. However, with $\beta$ fixed at $10$, Fig. \ref{range_fig:range_aloha_fading}(b) shows that, as $\alpha$ increases and approaches infinity, the optimum transmission range becomes equal under the {\em no fading} and {\em Rayleigh fading} channel models. Similarly, Fig. \ref{range_fig:op_rp_aloha} shows the optimized minimum number of transmissions required to deliver a packet to its destination located at unit distance from the source. The results in both figures show that though fading has an impact on the following parameters: optimum transmission range and network capacity with slotted ALOHA however, the impact is minimal on the optimized values of these parameters.

\subsection{Grid Pattern Based Medium Access Control Schemes}
\label{sec:impact_fading_grids}

We discussed earlier that under {\em Rayleigh fading}, the reception area of a transmitter in a grid pattern based scheme cannot be defined deterministically and we cannot deterministically evaluate the optimum transmission range as well. The purpose of this section is to discuss the impact of fading on the transmission range of grid pattern based schemes. Therefore, we will set a simplified framework for our analysis and use this framework to compare the transmission range of a transmitter under {\em no fading} and {\em Rayleigh fading} channel models. 

For our discussion in this section, we will consider {\em w.l.o.g.} the transmitter $i$ located at the origin.

\begin{theorem}
In case of grid pattern based schemes, under no fading channel model, the probability of successfully transmitting a packet from a transmitter i to a receiver located at distance r at point  ${\bf z}$ is
\begin{equation}
p(\lambda,r,\beta,\alpha)=H\left[\frac{r^{-\alpha}}{\beta}-I\right]~,
\label{range_eq:prob_grid_nofade}
\end{equation}
where $H[.]$ is the Heaviside function, $r=\|{\bf z}-{\bf z}_i\|$ and $I=\sum\limits_{j\neq i,{\bf z}_j\in {\cal S}}\frac{1}{\|{\bf z}-{\bf z}_j\|^{\alpha}}$. 
\end{theorem}
\begin{proof}
The proof of this is rather intuitive. Under {\em no fading} channel model, if the receiver at point ${\bf z}$ lies within the deterministic ${\cal A}_i({\cal S},\beta,\alpha)$, it will always receive the packet with probability one under the given orientation of the grid pattern and {\em no fading} because (using \eqref{schemes_eq:sinr_condition} with $N_0=0$ and $P_i=1$ (for \mbox{$i=0,1,\ldots,k,\ldots$}))
$$
\frac{1}{\|{\bf z}-{\bf z}_i\|^{\alpha}}\geq \beta \sum\limits_{j\neq i,{\bf z}_j\in {\cal S}}\frac{1}{\|{\bf z}-{\bf z}_j\|^{\alpha}}~.
$$
\end{proof}

In case of {\em Rayleigh fading} channel model, we can write \eqref{schemes_eq:sinr_condition} with $N_0=0$ and $P_i=1$ (for {$i=0,1,\ldots,k,\ldots$}) as
$$
\frac{F_i}{\|{\bf z}-{\bf z}_i\|^{\alpha}}\geq \beta \sum\limits_{j\neq i,{\bf z}_j\in {\cal S}}\frac{F_j}{\|{\bf z}-{\bf z}_j\|^{\alpha}}~.
$$
We recall that $F_k$ is an {\em i.i.d.} random variable but, in the following discussion, we will keep the subscripts for the purpose of clarification. We consider that the variations in the signal level cause by {\em Rayleigh fading} are viewed as {\em i.i.d.} variations in their respective transmit power level which also follow the distribution of the random variable $F\equiv F_k$.

\begin{theorem}
In case of grid pattern based schemes, under Rayleigh fading channel model, the probability of successfully transmitting a packet from a transmitter i to a receiver located at distance r at point  ${\bf z}$ is
\begin{equation}
p(\lambda,r,\beta,\alpha)=\prod\limits_{j\neq i,j\in{\cal S}}\left(\frac{1}{1+\beta\frac{\|{\bf z}-{\bf z}_j\|^{-\alpha}}{\|{\bf z}-{\bf z}_i\|^{-\alpha}}}\right)~.
\label{range_eq:prob_grid_fading}
\end{equation}
\end{theorem}
\begin{proof}
We know that
\begin{align*}
p(\lambda,r,\beta,\alpha)&=\Pr\left(\frac{F_i}{\|{\bf z}-{\bf z}_i\|^{\alpha}}\geq \beta \sum\limits_{j\neq i,{\bf z}_j\in {\cal S}}\frac{F_j}{\|{\bf z}-{\bf z}_j\|^{\alpha}}\right)\\
&=\Pr\left(\frac{F_i}{\sum\limits_{j\neq i,{\bf z}_j\in {\cal S}}F_j\frac{\|{\bf z}-{\bf z}_j\|^{-\alpha}}{\|{\bf z}-{\bf z}_i\|^{-\alpha}}}\geq \beta\right)\\
&=\Pr\left(\frac{F_i}{\sum\limits_{j\neq i,{\bf z}_j\in {\cal S}}w_j F_j}\geq \beta\right)~,\\
\intertext{where $w_j=\frac{\|{\bf z}-{\bf z}_j\|^{-\alpha}}{\|{\bf z}-{\bf z}_i\|^{-\alpha}}$ and $$G=\sum\limits_{j\neq i,{\bf z}_j\in {\cal S}}w_j F_j~,$$ is the {\em weighted} interference which is expressed as the weighted sum of independent random variables $F_j$.}
\intertext{Let us denote the Moment Generating Function (MGF) of $G$ by $M_G(s)$ and, using the definition of MGF of {\em Rayleigh fading}, we have $$M_G(s)=\prod\limits_{j\neq i,{\bf z}_j\in {\cal S}}M_{F_j}(w_j s)=\prod\limits_{j\neq i,{\bf z}_j\in {\cal S}}\E[e^{w_j s F_j}]=\prod\limits_{j\neq i,{\bf z}_j\in {\cal S}}\frac{1}{1-w_j s}~.$$} 
\intertext{Therefore, we can derive the probability of successful reception from transmitter $i$ at receiver located at distance $r$ at point ${\bf z}$ as}
p(\lambda,r,\beta,\alpha)&=\Pr(F_i\geq \beta G)\\
&=\int\limits_{x=0}^\infty \int\limits_{y=0}^\infty \phi_F(x) \phi_G(y).\mathds{1}_{x\geq \beta y} dx dy~, \\
\intertext{where $\phi_F(x)$ and $\phi_G(y)$ are the probability density functions of the random variables $F\equiv F_i$ and $G$ respectively. Therefore, we have}
p(\lambda,r,\beta,\alpha)&=\int\limits_{y=0}^\infty \int\limits_{x=\beta y}^\infty \phi_F(x) dx \phi_G(y) dy \\
&=\int\limits_{y=0}^\infty \exp(-\beta G) \phi_G(y) dy \\
&=\E[\exp(-\beta G)]\\
&=M_G(-s)\Big|_{s=\beta}=\prod\limits_{j\neq i,{\bf z}_j\in {\cal S}}\frac{1}{1+w_j \beta}~,
\end{align*}
which completes the proof.
\end{proof}

Now we will use our results to compute the probability of successful transmission with grid pattern based schemes under {\em no fading} and {\em Rayleigh fading} channel models. {\em W.l.o.g.}, we will only consider the transmitter $i$ located at the origin and we will use \eqref{range_eq:prob_grid_nofade} and \eqref{range_eq:prob_grid_fading} to compute the probability of successful reception when the receiver is at distance $r$ from transmitter $i$. We will use the same network setting as described in \S \ref{range_sec:grid_pattern} and we will use the square grid pattern based scheme. We assume that the set ${\cal S}$ has density $\lambda=1$. Therefore, $d=\sqrt{1/\lambda}=1$. 

We consider that point ${\bf z}$ lies on the line joining $(0,0)$ and $(1,1)$ and we numerically compute the probability of successful reception at distance $r$ from transmitter $i$ using \eqref{range_eq:prob_grid_nofade} and \eqref{range_eq:prob_grid_fading} as point ${\bf z}$ moves on the straight line, from $(0,0)$ to $(1,1)$. Figure \ref{range_fig:dist_grids} shows the comparison of $p(1,r,\beta,\alpha)$ with square grid pattern based scheme under {\em no fading} and {\em Rayleigh fading}. In contrast to the performance of slotted ALOHA in Fig. \ref{range_fig:dist_aloha_fading}, we notice that fading can significantly degrade the performance of grid pattern based schemes.

\begin{figure}[!t]
\centering
\psfrag{r}{$r$}
\includegraphics[scale=0.9]{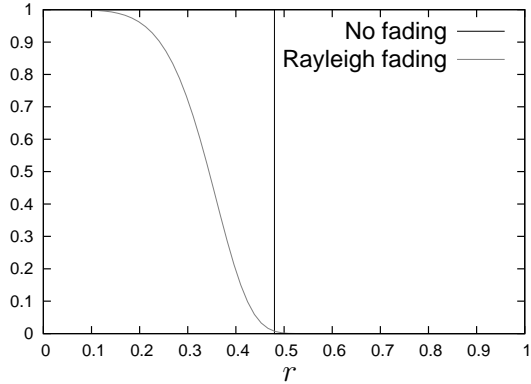}
\caption{\mbox{$p(1,r,\beta,\alpha)$} versus $r$ of square grid pattern based scheme under {\em no fading} and {\em Rayleigh fading} channel models.  \mbox{$\lambda=1$}, $\beta=1.0$ and \mbox{$\alpha=4.0$}.
\label{range_fig:dist_grids}}
\end{figure}

\section{Conclusions}
\label{sec:conclude}

We studied the problem of optimizing the capacity of multi-hop wireless network via geometric analysis of various MAC schemes. We evaluated the performance of network under the framework of normalized optimum transmission range and computed the optimized number of retransmissions required to transport a packet over multi-hop path to its final destination. We used analytical tools, based on realistic interference model, to evaluate the performance of slotted ALOHA and grid pattern based schemes. Our results show that designing of the MAC scheme to optimize these quantities shall take into account the system parameters like SIR threshold and attenuation coefficient. Our results also show that under {\em no fading}, at typical values of SIR threshold equal to $10.0$ and attenuation coefficient equal to $4.0$, most optimal scheme is based on triangular grid pattern and compared to slotted ALOHA, which does not use any significant protocol overhead, triangular grid pattern based scheme can only improve the normalized optimum transmission range and network capacity by factors of {\em two} and {\em three} respectively. We also studied the performance of these schemes under {\em Rayleigh fading} and showed that this gap in performance is even closer. As a future work, it will also be interesting to see how does the performance of node coloring and CSMA schemes compare with the performance of the MAC schemes discussed in this article. The conclusions of this work is that improvements above ALOHA are limited in performance and maybe costly in terms of protocol overheads. Our work also promotes the idea of cross-layer optimization between MAC and routing layers and encourages that attention should be focused on optimizing existing medium access schemes and studying, designing and implementing new efficient routing strategies, {\em e.g.}, as discussed in \cite{weber08,gomez-cambell2004}. 

\bibliographystyle{hieeetr}
\bibliography{Bibliography}

\end{document}